%% file: main.tex
\documentclass[acmsmall]{acmart}

\usepackage[nolist,nohyperlinks]{acronym}
\usepackage{algpseudocode}
\usepackage{algorithm}
\usepackage{adjustbox}
\usepackage{amsmath}
\usepackage{pgfplots}
\pgfplotsset{compat=1.18}
\usepgfplotslibrary{fillbetween}

\AtBeginDocument{%
  \providecommand\BibTeX{{%
    \normalfont B\kern-0.5em{\scshape i\kern-0.25em b}\kern-0.8em\TeX}}}

\setcopyright{acmcopyright}
\copyrightyear{2024}
\acmYear{2024}
\acmDOI{XXXXXXX.XXXXXXX}

\theoremstyle{definition}
\newtheorem{definition}{Definition}

\begin{document}

\title{Cross-Blockchain Communication Using Oracles With an Off-Chain Aggregation Mechanism Based on zk-SNARKs}

\author{Michael Sober}
\email{michael.sober@tuhh.de}
\orcid{0000-0002-9612-9022}
\affiliation{%
  \institution{Hamburg University of Technology}
  \streetaddress{Blohmstr. 15}
  \city{Hamburg}
  \country{Germany}
  \postcode{21079}
}
\affiliation{%
  \institution{Christian Doppler Laboratory for Blockchain Technologies for the Internet of Things}
  \streetaddress{Blohmstr. 15}
  \city{Hamburg}
  \country{Germany}
  \postcode{21079}
}

\author{Giulia Scaffino}
\email{giulia.scaffino@tuwien.ac.at}
\orcid{0000-0001-5680-3003}
\affiliation{%
  \institution{TU Wien}
  \streetaddress{Favoritenstr. 9}
  \city{Vienna}
  \country{Austria}
  \postcode{1040}
}
\affiliation{%
  \institution{Christian Doppler Laboratory for Blockchain Technologies for the Internet of Things}
  \streetaddress{Blohmstr. 15}
  \city{Hamburg}
  \country{Germany}
  \postcode{21079}
}

\author{Stefan Schulte}
\email{stefan.schulte@tuhh.de}
\orcid{0000-0001-6828-9945}
\affiliation{%
  \institution{Hamburg University of Technology}
  \streetaddress{Blohmstr. 15}
  \city{Hamburg}
  \country{Germany}
  \postcode{21079}
}
\affiliation{%
  \institution{Christian Doppler Laboratory for Blockchain Technologies for the Internet of Things}
  \streetaddress{Blohmstr. 15}
  \city{Hamburg}
  \country{Germany}
  \postcode{21079}
}

\renewcommand{\shortauthors}{Sober et al.}

\input{00_abstract}



\received{20 February 2007}
\received[revised]{12 March 2009}
\received[accepted]{5 June 2009}

\maketitle

\input{acronyms.tex}

\acresetall
\acused{zk-SNARK}

\input{01_introduction}
\input{02_background}
\input{03_design}
\input{04_security}
\input{05_implementation}
\input{06_evaluation}
\input{07_related_work}
\input{08_conclusion}
\input{09_acks}

\bibliographystyle{ACM-Reference-Format}
\bibliography{refs}

\end{document}

%% file: 00_abstract.tex
\begin{abstract}

The closed architecture of prevailing blockchain systems renders the usage of this technology mostly infeasible for a wide range of real-world problems. Most blockchains trap users and applications in their isolated space without the possibility of cooperating or switching to other blockchains. Therefore, blockchains need additional mechanisms for seamless communication and arbitrary data exchange between each other and external systems. Unfortunately, current approaches for cross-blockchain communication are resource-intensive or require additional blockchains or tailored solutions depending on the applied consensus mechanisms of the connected blockchains. Therefore, we propose an oracle with an off-chain aggregation mechanism based on \ac{zk-SNARKs} to facilitate cross-blockchain communication. The oracle queries data from another blockchain and applies a rollup-like mechanism to move state and computation off-chain. The \textit{zkOracle} contract only expects the transferred data, an updated state root, and proof of the correct execution of the aggregation mechanism. The proposed solution only requires constant 378 kgas to submit data on the Ethereum blockchain and is primarily independent of the underlying technology of the queried blockchains.

\end{abstract}

%% file: acronyms.tex
\begin{acronym}
    \acro{EVM}{Ethereum Virtual Machine}
    \acro{ZKP}{Zero-Knowledge Proof}
    \acro{DKG}{Distributed Key Generation}
    \acro{VSS}{Verifiable Secret Sharing}
    \acro{PVSS}{Publicly Verifiable Secret Sharing}
    \acro{zk-SNARK}{Zero-Knowledge Succinct Non-interactive Arguments of Knowledge}
    \acro{zk-SNARKs}{Zero-Knowledge Succinct Non-interactive Arguments of Knowledge}
    \acro{zk-STARKs}{Zero-Knowledge Succinct (Scalable) Transparent Arguments of Knowledge}
    \acro{R1CS}{Rank-1 Constraint-System}
    \acro{QAP}{Quadratic Arithmetic Program}
    \acro{VOC}{Verifiable Off-Chain Computation}
    \acro{EVM}{Ethereum Virtual Machine}
    \acro{NIZK}{Non-interactive Zero-Knowledge Proof}
    \acro{TID}{Threshold Information Disclosure}
    \acro{CP-ABE}{Ciphertext-Policy Attribute-Based Encryption}
    \acro{PoS}{Proof of Stake}
    \acro{DPoS}{Delegated Proof of Stake}
    \acro{PoW}{Proof of Work}
    \acro{MPC}{Multiparty Computation}
    \acro{DSL}{Domain-Specific Language}
    \acro{CRS}{Common Reference String}
    \acro{IBC}{Inter-blockchain Communication Protocol}
    \acro{XCM}{Cross-Consensus Message Format}
    \acro{CCIP}{Cross-Chain Interoperability Protocol}
    \acro{OCR}{Off-chain Reporting}
    \acro{dApp}{Decentralized Application}
    \acro{SPV}{Simplified Payment Verification}
    \acro{EdDSA}{Edwards-curve Digital Signature Algorithm}
    \acro{ETH}{Ether}
    \acro{OCR}{Off-Chain Reporting}
    \acro{EUF-CMA}{Existential Unforgeability under Chosen Message Attack}
\end{acronym}

%% file: 01_introduction.tex
\section{Introduction}
\label{sec:introduction}


Blockchains provide a distributed ledger of transactions while ensuring the integrity, consistency, and immutability of the transactions. Blockchain technology has been first introduced with Bitcoin~\cite{nakamoto2008bitcoin} and is now considered the main backing technology of cryptocurrencies. However, it also found application in various other application areas like healthcare~\cite{agbo2019blockchain}, supply chain management~\cite{saberi2019blockchain}, and others~\cite{dai2019blockchain, xie2019survey, demestichas2020blockchain}. The second generation of blockchains, like Ethereum~\cite{wood2014ethereum}, provides smart contracts, which are deterministic programs stored on the blockchain executed by every node of the network. The application of smart contracts allows the creation of decentralized applications~\cite{cai2018decentralized}, opening an even wider range of possible application areas.


Different application domains of blockchain technology come with different challenges and requirements. A single blockchain can only fulfill some criteria, which results in a wide range of solutions tailored to particular application domains. These development efforts lead to a highly fragmented and heterogenous blockchain space~\cite{schulte2019towards}, as different blockchains often have unique consensus mechanisms, token standards, or smart contract capabilities. These broad differences make it difficult for different blockchains to work together and require special attention during development to ensure an interoperable blockchain landscape. Unfortunately, blockchain interoperability is a major issue often overseen and hinders blockchain technology from developing its full potential. Blockchains are usually considered closed systems that can not interact with external systems. While exceptions exist, e.g., Hyperledger Cacti~\cite{hyperledger2022cacti}, Cosmos~\cite{kwon2020cosmos}, Polkadot~\cite{wood2016polkadot}, and others~\cite{ellul2022verifiable}, missing interoperability is still a widespread problem. Amongst other things, missing interoperability between blockchains leads to vendor lock-in, i.e., users can not easily migrate to other blockchains~\cite{belchior2023you}. 

Further, smart contracts can only manipulate their state on the hosting blockchain and can not interact with smart contracts deployed on other blockchains or external systems. Therefore, there is a strong need for blockchain interoperability solutions to break the barrier between heterogeneous blockchain systems and enable cross-blockchain communication to securely transfer arbitrary state information between blockchains in a decentralized way. The possibility to transfer arbitrary information between blockchains enables the creation of decentralized cross-blockchain applications and protocols~\cite{nissl2021crossblockchain, gpact2021robinson, scaffino2024alba}.


There is an ongoing effort in research and industry to create blockchain interoperability solutions~\cite{belchior2021survey,wang2023interoperability}. However, current solutions, like blockchain relays, are often very complex to implement since they depend on the source blockchains' consensus mechanism. A blockchain of blockchains solution has even more technical complexity, requiring an additional blockchain for cross-blockchain communication. On the other hand, notary schemes offer great flexibility and are mostly independent of the connected blockchains. Unfortunately, notary schemes exhibit a lower degree of decentralization~\cite{wang2023interoperability}. Many notary schemes follow a decentralized approach where $n$ out $m$ parties must collaborate to distribute trust among multiple parties and reach an agreement, e.g., by applying a voting mechanism. These voting mechanisms require costly on-chain interactions since the smart contract has to process many votes to determine the outcome. More cost-efficient solutions use complex off-chain voting mechanisms, e.g., by applying threshold signatures, to submit an aggregated result~\cite{sober2021voting}. However, these solutions require the additional execution of \ac{DKG} protocols which introduce additional complexity and security risks~\cite{sober2023distributed}.

Therefore, we propose a solution for cross-blockchain communication by using a decentralized oracle, i.e., a bridge to an external data source, with an off-chain aggregation mechanism based on \ac{zk-SNARKs}. The oracle allows its clients to query arbitrary information from other blockchains independent of the inner workings of the source blockchain. A committee of oracle nodes, divided into a single aggregator and multiple validators, executes an off-chain aggregation mechanism to attest to the correctness of the retrieved information. The validators query the source blockchain for the requested information and send the signed result to the current aggregator. The aggregator combines the results by generating a \ac{zk-SNARK} to verify the signatures and distribute the rewards off-chain. Afterward, the aggregator submits the proof, the updated state, and the queried information to the smart contract. The smart contract must only verify the proof, update the state, and store the retrieved data. Further, we provide a prototypical implementation for Ethereum-based blockchains and evaluate it concerning costs, memory, and performance. The results show that our oracle only requires 378~kgas to submit data on Ethereum-based blockchains. Furthermore, oracle nodes without many resources can also use the system and efficiently support a committee size of 256. 


The remainder of this paper is structured as follows: Section~\ref{sec:background} provides general information about the applied concepts. Afterward, Section~\ref{sec:system_design} introduces the design of the oracle with its off-chain aggregation mechanism, and Section~\ref{sec:security} follows with a security analysis. Section~\ref{sec:implementation} provides implementation details, and Section~\ref{sec:evaluation} evaluates the proposed solution concerning costs, memory, and performance. After that, Section~\ref{sec:related_work} follows with a discussion of related work. Finally, Section~\ref{sec:conclusion} concludes the paper.  

%% file: 02_background.tex
\section{Background}
\label{sec:background}

This section discusses the underlying concepts of the proposed solution. We provide the basics of cross-blockchain communication and follow with an explanation of \acp{ZKP} and rollups.

\subsection{Cross-Blockchain Communication}

Cross-blockchain communication refers to the ability of two blockchains to reach a consistent state across both blockchains by synchronizing transaction execution~\cite{wang2023interoperability}. Hence, cross-blockchain communication allows distinct blockchains to transfer arbitrary information, query the state of another blockchain~\cite{sober2021voting} or allow more specific use cases such as atomic swaps~\cite{herlihy2018atomic} and cross-blockchain smart contract calls~\cite{nissl2021crossblockchain}. Unfortunately, cross-blockchain communication requires a trusted third party~\cite{zamyatin2021sok} since it is otherwise infeasible for two blockchains to verify each other’s state. However, the ability to verify the state of other blockchains is a requirement to allow for interoperable blockchains to exist~\cite{lafourcade2020blockchain}. The trusted third party must attest to the correct execution of a transaction to another blockchain, allowing the interdependent execution of transactions on multiple blockchains. These trusted third parties can both be centralized or decentralized and most prominently include blockchain relays, notary schemes, and blockchains of blockchains~\cite{belchior2021survey}.

A blockchain relay is a smart contract deployed on a target blockchain to replicate the source blockchain~\cite{belchior2021survey}. Off-chain clients continuously relay block headers from the source blockchain to the target blockchain while the smart contract enforces the consensus rules of the source blockchain to rebuild the chain. The off-chain clients relaying the block headers do not need to be trusted since the smart contract enforces the correct behavior. Further, the smart contract offers light client functionalities on the target blockchain, allowing smart contracts to query the state of the source blockchain. Typically, blockchains store a Merkle root of a block's transactions in the block header. Hence, the smart contract allows the execution of \ac{SPV}, which allows checking if a particular transaction is part of a block while only having the block header. For that, clients submit a Merkle proof of inclusion to the relay contract to show that a transaction is part of a specific block. Further, the relay contract checks that the block containing the transaction is included in the source blockchain and considered final. Blockchain relays have a high degree of decentralization but are costly since on-chain block header verification and storage consume many resources. Additionally, off-chain clients must always keep the relay in sync with the source blockchain requiring considerable maintenance effort. However, novel relay solutions already apply techniques to move computation off-chain by either optimistically accepting new block headers~\cite{frauenthaler2020eth} or using \acp{ZKP}~\cite{xie2022zkbridge,westerkamp2020zkrelay} to achieve better cost efficiency.

Another technique for cross-blockchain communication is through the application of notary schemes~\cite{buterin2016chain}. A notary is a trusted entity or group  responsible for executing certain operations across the boundaries of blockchains. The notary can monitor the involved blockchains, read the blockchains' state, and publish new transactions. Therefore, the notary can function as an oracle to query state information or publish new transactions on behalf of another blockchain. The blockchains interacting with the notary must only authenticate the notary to verify the integrity and authenticity of the executed operations. Notary schemes are a relatively simple solution for cross-blockchain communication since they are independent of the underlying blockchains and offer great flexibility. However, notary schemes exhibit a lower degree of decentralization. Trusting only a single entity leads to a single point of failure, which can compromise the whole system. Therefore, many solutions require that $n$ out of $m$ entities agree by executing a voting mechanism, either on-chain or off-chain, to distribute the trust among multiple entities~\cite{sober2021voting}. Further, these approaches require additional mechanisms to form a committee of trustworthy entities and reach an agreement.

Multiple blockchains can also communicate by following a blockchain of blockchains approach~\cite{wang2023interoperability}. A blockchain of blockchains provides a platform of independent subchains to communicate with each other through a central relay chain. The main chain is responsible for recording the connected subchains' actions. Every subchain can access the main chain to verify the activities of other subchains to implement cross-blockchain communication. The subchains can work independently of the main chain~\cite{kwon2020cosmos} or adopt a shared security approach where validators of the main chain also validate the subchains~\cite{wood2016polkadot}. The main chain is, to a certain extent, a notary scheme with the properties of a blockchain, as it attests to the actions of the subchains. However, in many cases, it also fulfills other purposes besides cross-blockchain communication. Unfortunately, there is also a considerable overhead to using an additional blockchain to keep track of all connected chains.

\subsection{Zero-Knowledge Proofs}

\acp{ZKP} were first introduced by Goldwasser et al.~\cite{goldwasser2019knowledge} to enable a prover to convince a verifier that the prover knows a witness to a particular statement without revealing any additional knowledge. For that, a \ac{ZKP} system has to satisfy the following three properties~\cite{groth2016size}:
\begin{enumerate}
\item \textit{Completeness:} An honest prover can always convince an honest verifier that the prover knows a witness to a true statement.
\item \textit{Soundness:} A malicious prover can not convince an honest verifier about the knowledge of a witness to a false statement.
\item \textit{Zero-knowledge:} A verifier does not learn additional information besides the statement's validity.
\end{enumerate}

However, these proof systems require multiple interactions between the prover and the verifier to communicate the proof and convince the verifier.  Therefore, Blum et al.~\cite{blum2019non} introduced non-interactive ZKP systems as an alternative to interactive proof systems.  These non-interactive proof systems require an initial setup phase in which a prover and a verifier generate and share a common reference string that replaces the previously necessary interactions. After that, a prover can generate a \ac{NIZK} to convince the verifier, while the verifier does not have to send any messages to the prover. Further, Goldreich et al.~\cite{goldreich1991proofs} showed that all statements in $\mathcal{NP}$ have \acp{ZKP}, enabling the possibility of verifying arbitrary computations. Verifying computational integrity allows a verifier to check whether the output of a given program, given its inputs, is correct without re-executing the program itself~\cite{gennaro2010non, parno2012delegate}. These properties are particularly advantageous in decentralized systems (e.g., blockchains) where provers cannot directly communicate with verifiers (e.g., smart contracts) or want to convince several verifiers.

Nowadays, \acp{ZKP} are heavily applied in blockchain technology to create scalable~\cite{thibault2022rollups} and privacy-aware solutions~\cite{almashaqbeh2022sok}. Blockchains require every node in the network to execute and verify transactions which constitutes a major scalability issue. However, applying \acp{ZKP} makes it possible to move computations off-chain such that nodes only need to verify small proofs, reducing the overall network load and increasing throughput~(see~Section~\ref{subsec:rollups}). Further, most blockchains only provide pseudonymity instead of anonymity, allowing anyone to link multiple addresses to the same identity by applying heuristics for address clustering~\cite{zhang2020heuristic}. With the application of \acp{ZKP}, it is possible to break the link between multiple addresses and create privacy-preserving blockchains~\cite{sasson2014zerocash,hopwood2016zcash}.

The main \ac{ZKP} proof systems applied are \ac{zk-SNARKs}~\cite{ben2014succinct}, \ac{zk-STARKs}~\cite{ben2019scalable}, and Bulletproofs~\cite{bunz2018bulletproofs}. These proof systems differ in their underlying cryptographic assumptions, complexity, and the necessity of a trusted setup. The basic cryptographic assumptions behind \ac{zk-SNARKs} are the hardness assumption of the elliptic curve discrete logarithm problem and the existence of secure bi-linear pairings. While the proving complexity of \ac{zk-SNARKs} is quasi-linear, the verification and communication complexity is constant. Unfortunately, \ac{zk-SNARKs} require a trusted setup phase to generate a \ac{CRS}. The construction of \ac{zk-STARKs}, on the other hand, is based upon collision-resistant hash functions, requiring no trusted setup and having quasi-linear proving complexity but a poly-logarithmic verification and communication complexity. Finally, Bulletproofs rely on the discrete log assumption and require no trusted setup but have quasi-linear proving complexity, linear verification complexity, and logarithmic communication complexity. In the work at hand, we use \ac{zk-SNARKs} to take advantage of the constant verification and communication complexity, which is particularly beneficial for verifying proofs with smart contracts.

The first protocol enabling efficient \ac{zk-SNARK} constructions was the Pinocchio protocol~\cite{parno2016pinocchio, ben2014succinct}. After that, protocols like Gepetto~\cite{costello2015geppetto} and Groth16~\cite{groth2016size} further improved the efficiency of \ac{zk-SNARKs}, and recent works introduced \ac{zk-SNARKs} with a universal setup~\cite{gabizon2019plonk, maller2019sonic, chiesa2020marlin}. However, Groth16 is still heavily used in blockchain technology due to its small proof size and low verification complexity. Hence, we apply Groth16 to create the off-chain aggregation mechanism. The protocol allows the creation of zk-SNARKs for arithmetic circuits. For that, a program is translated into an arithmetic circuit, which is converted to a \ac{R1CS}. Finally, the \ac{R1CS} is converted to a \ac{QAP}, representing the program in a different format to generate the proof. The input to the program also yields the corresponding witness to the \ac{QAP}. After that, the setup produces a \ac{CRS} and a simulation trapdoor. The simulation trapdoor needs to be forgotten because it allows the creation of fake proofs. However, by using \ac{MPC}, it is possible to distribute the trust among multiple parties~\cite{bowe2017scalable}. Finally, the algorithm can use the \ac{CRS} to create pairing-based zk-SNARKs for the witness to the \ac{QAP}.

\subsection{Rollups}
\label{subsec:rollups}

Blockchain technology faces a considerable scalability issue since blockchains can not cope with the growing amount of work~\cite{zhou2020solutions}. Research has already come up with various solutions, including improvements to Layer-1, i.e., the base layer of the protocol, and Layer-2 protocols which build on top of a blockchain~\cite{hafid2020scaling}. Rollups are a Layer-2 blockchain scaling solution to increase the throughput of blockchains by moving transaction computation and state off-chain~\cite{thibault2022rollups}. Aggregators execute multiple transactions off-chain and only submit state changes and a bundle of compressed transaction data to the blockchain, i.e., the nodes on Layer-1 solely need to execute a single transaction. A smart contract manages the state root of the rollup and allows anyone to submit a bundle of compressed transactions, the previous and the new state root. The corresponding Merkle tree to the state root contains the account data (e.g., balance, nonce) of the different accounts which are part of the rollup. Users can deposit and withdraw funds directly by interacting with the smart contract.

Rollups have to enact specific rules to determine the validity of the submitted transaction data and the according state changes. Currently, two main types of rollups exist: optimistic and zk-rollups~\cite{buterin2021rollups}. Optimistic rollups~\cite{kalodner2018arbitrum} assume that transactions are valid and allow users to submit fraud proofs~\cite{al2018fraud} during a dispute period to prove a transaction is invalid. The smart contract verifies the proof and reverts the rollup to its previous state. With zk-rollups, on the other hand, users submit validity proofs (e.g., zk-SNARKs) to prove the correctness of the new state root resulting from the execution of the transactions in the bundle. Therefore, zk-rollups do not require a dispute period and enable instant transaction finality.

The availability of the compressed transaction data on Layer-1 ensures that anyone can recompute and verify the state of the rollup. While zk-rollups only require minimal transaction data on Layer-1 to ensure that anyone can recompute and verify the state, optimistic rollups require additional data to enable the verification of fraud proofs. Therefore, even with malicious actors in the system, rollups can provide censorship resistance and liveness.

Moving transaction computation and state off-chain is a promising approach to improve aggregation mechanisms that require many interactions with a smart contract or other tools (e.g., threshold signatures) to reduce on-chain computation. Hence, in the following, we apply the concept behind zk-rollups to create an off-chain aggregation mechanism.

%% file: 03_design.tex
\section{System Design}
\label{sec:system_design}

In this section, we propose a cross-blockchain communication mechanism using oracles with an off-chain aggregation mechanism based on \ac{zk-SNARKs}. Initially, we provide a short overview of the system. Afterward, we discuss dynamic participation and the applied incentive mechanism. Finally, we define the off-chain aggregation mechanism.

\subsection{Overview}
\label{subsec:overview}

\begin{figure}[t]
  \centering
  \begin{adjustbox}{width=\linewidth}
  \includegraphics[width=\linewidth]{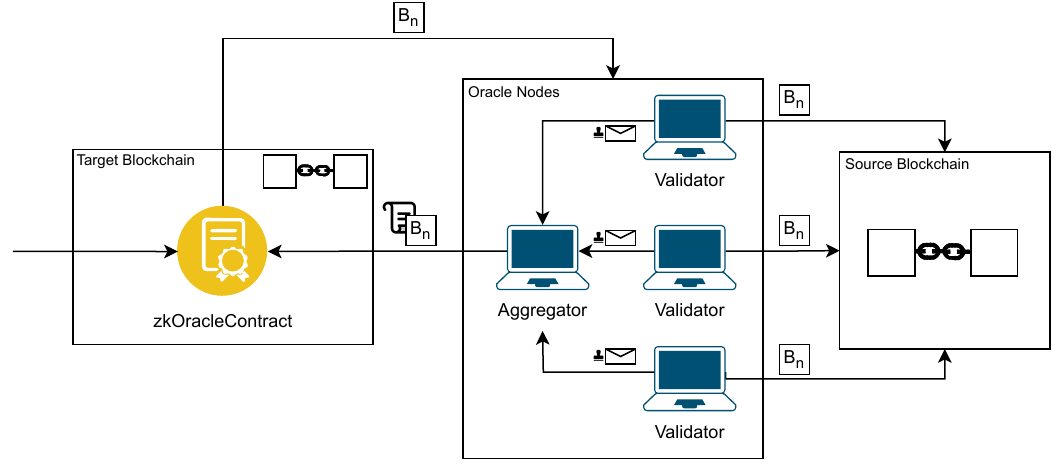}
  \end{adjustbox}
  \caption{Overview of the system}
  \label{fig:system_overview}
\end{figure}

An oracle is a bridge between a blockchain and an external data source. The typical architecture for an oracle system consists of an on-chain component, e.g., a smart contract, and an off-chain component. The off-chain component retrieves data from external data sources, e.g., another blockchain, ensures the data's integrity, and submits the data to the oracle contract~\cite{heiss2019oracles}. The oracle contract receives the data and serves it to data consumers.  There are different types of oracles based on the data source, trust level, communication direction, and data type~\cite{al2020trustworthy}. However, this work focuses on cross-chain oracles to enable blockchain interoperability and data transfer between distinct blockchains. These oracles represent a flexible and lightweight solution for cross-blockchain communication as they do not depend on the source blockchain's consensus mechanism and require no additional blockchain. Hence, the proposed oracle solution is compatible with different blockchains, assuming the connected blockchains support smart contracts.

The oracle~(see~Fig.~\ref{fig:system_overview}) allows arbitrary data transfer between two blockchains by enabling the oracle's clients to issue requests to query data from another blockchain. The system consists of the \textit{zkOracle} contract deployed on a target blockchain and off-chain oracle nodes, which query data from a source blockchain and attest to the data's integrity. The \textit{zkOracle} contract offers a mechanism for dynamic participation~(see~Section~\ref{subsec:dynamic_participation}) that allows anyone to register and join a committee of oracle nodes. This committee collectively responds to requests by querying another blockchain and executing an off-chain aggregation mechanism~(see~Section~\ref{subsec:off_chain_aggregation}) before returning the data. The oracle selects a single aggregator responsible for aggregating the results from the committee members. The other committee members act as validators by querying the data from the source blockchain and attesting to the correctness. The aggregator collects a majority of answers with the same result and creates an aggregated result in a rollup-like fashion. For that, the aggregator generates a \ac{zk-SNARK} proving that the majority voted for a specific answer and that the rewards are distributed correctly. By generating this proof, the oracle can execute the voting mechanism and reward distribution off-chain, while the smart contract only needs to verify the proof to ensure the correctness of these computations.

Additionally, the oracle applies an incentive mechanism~(see~Section~\ref{subsec:incentive_mechanism}) to reward honest behavior and punish misbehavior. While honest nodes participating in the aggregation protocol receive a constant reward to compensate them for providing resources, misbehaving nodes may get slashed and lose part or all of their deposited collateral. The aggregator slashes misbehaving oracle nodes by generating a \ac{zk-SNARK} that the oracle node in question provided an answer that differs from the majority. Upon receiving the aggregated result, the \textit{zkOracle} contract only verifies the \ac{zk-SNARK}, stores the result, and updates the Merkle root containing the state of all oracle nodes. Afterward, any other smart contract can access the data retrieved from the source blockchain by calling the \textit{zkOracle} contract.

\subsection{Dynamic Participation}
\label{subsec:dynamic_participation}

The oracle relies on a committee of oracle nodes with a maximum size of $n$ nodes to collectively return the result to a query. However, the oracle can already operate with less than $n$ nodes as long as $\frac{n}{2}+1$ nodes are honest. For that, the system provides a mechanism for dynamic participation, allowing anyone to register oracle nodes and participate in transferring data between two blockchains. The mechanism enables any oracle node to join or leave the system by calling the \textit{zkOracle} contract deployed on the target blockchain. Further, the mechanism provides high Sybil resistance by following a \ac{PoS} mechanism to select the committee. While other mechanisms like using a pre-defined set of participants or a first come, first served approach are also possible, they suffer from Sybil attacks~\cite{douceur2002sybil} and a lower degree of decentralization.

Initially, the \textit{zkOracle} contract has no registered oracle nodes. Anyone can call the \textit{zkOracle} contract and register an oracle node providing a public key, IP address, and stake in the blockchain's native currency. The \textit{zkOracle} contract uses a Merkle Tree to store the accounts of registered oracle nodes consisting of an address, index, public key, and the stake balance. Therefore, the \textit{zkOracle} contract does not store raw account data, and the off-chain aggregation protocol can update the Merkle root to apply state changes resulting in considerable cost savings. The \textit{zkOracle} contract only allows up to $n$ oracle nodes providing a minimum stake to register. After the registration of $n$ oracle nodes, it is only possible for a new oracle node to join if it provides a higher stake than the oracle node it wishes to replace. The \textit{zkOracle} contract enables anyone to provide a public key, a higher stake than a specific oracle node, and the respective Merkle proof to replace the node. With this approach, the \textit{zkOracle} contract ensures that the committee of oracle nodes always consists of the $n$ nodes with the highest deposited stake.

Additionally, registered oracle nodes can also leave the system. Leaving the system requires two steps: exit and withdraw. First, an oracle node calls the \textit{zkOracle} contract to announce that it wants to exit the system by providing the necessary account details with the respective Merkle proof. The \textit{zkOracle} contract ensures that the caller is the rightful owner of the account, the validity of the Merkle proof, and computes the time an oracle node has to wait until it is allowed to withdraw. The exit time ensures that users with a high enough stake can not repeatedly remove other oracle nodes with a lower stake from the system by replacing the node and immediately leaving the system to replace other nodes. Without the exit time, anyone with a high enough stake could remove all nodes with a lower stake and replace these oracle nodes with new oracle nodes that do not require a high stake. After reaching the exit time, an oracle node can provide its account details with the corresponding Merkle proof to withdraw its stake from the \textit{zkOracle} contract. Finally, the \textit{zkOracle} contract transfers the stake and replaces the account in the Merkle tree with an empty account.

\subsection{Incentive Mechanism}
\label{subsec:incentive_mechanism}

The proposed oracle solution is a public system allowing anyone to register oracle nodes cooperating to transfer data between blockchains. However, without an incentive, there is no reason for anyone to participate in the system since participation requires resources. Therefore, the oracle applies an incentive mechanism that encourages participation and honest behavior and discourages malicious behavior. For that, every oracle node, whether acting as an aggregator or a validator, receives a monetary reward for submitting a result to the \textit{zkOracle} contract. On the other hand, the \textit{zkOracle} contract punishes misbehaving nodes by slashing the deposited stake.

The aggregator and validators receive a monetary reward for each submission. This reward can be constant, or the \textit{zkOracle} contract can dynamically define the amount. However, dynamically determining the reward requires additional public inputs to the aggregation circuit, which increases the verification costs. The oracle distributes these rewards during the off-chain aggregation mechanism. The aggregator increases its balance and the balance of every validator submitting a valid result and signature by the defined reward. The rewards of the aggregator need to be higher than the validator rewards since the aggregator has to provide a lot more computational power to generate the \ac{zk-SNARK} during the execution of the off-chain aggregation mechanism and also has to pay the transaction fees for submitting the final result to the \textit{zkOracle} contract. 

The aggregator assigns the rewards during off-chain aggregation and proves that it assigned the correct rewards to the involved participants with the generated \ac{zk-SNARK}. The \textit{zkOracle} contract only needs to verify the proof and store the updated Merkle root, which contains the account information with the updated balances. Therefore, the system provides a fair and cost-efficient approach to reward every participant without relying on multiple on-chain interactions or relying on game theoretical approaches to distribute the rewards.

An aggregator selects only a majority of responses to create aggregated results. This selection process allows it to potentially exclude a minority of validators from receiving rewards. As a result, resources are wasted even though the validators’ tasks are not very resource-intensive. Nevertheless, the oracle eventually appoints each validator as an aggregator, enabling them to earn rewards. This mechanism results in a trade-off between the fees clients pay to use the oracle and the fair distribution of rewards among validators since clients must pay for each validator submitting an answer.

\subsection{Off-Chain Aggregation}
\label{subsec:off_chain_aggregation}

\begin{figure}[t]
  \centering
  \begin{adjustbox}{width=0.8\linewidth}
  \includegraphics[width=\linewidth]{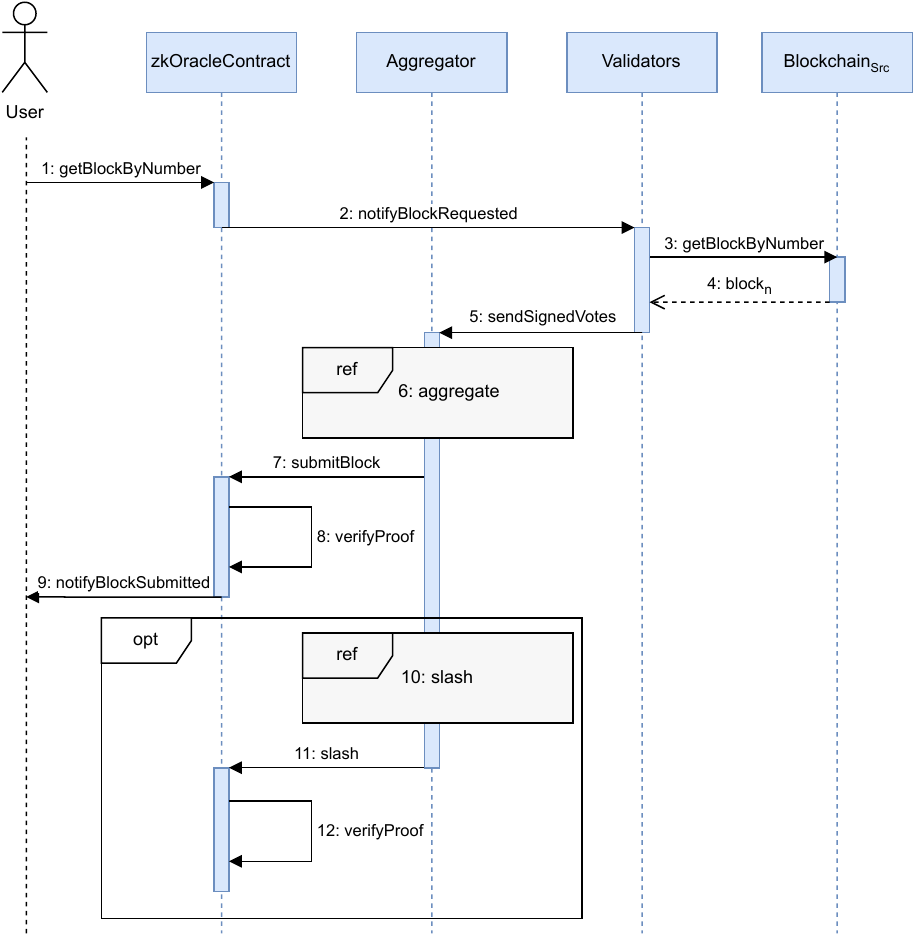}
  \end{adjustbox}
  \caption{Sequence of interactions for retrieving data}
  \label{fig:sequence}
\end{figure}

The off-chain aggregation mechanism~(see~Fig.~\ref{fig:sequence}) allows the oracle nodes to return aggregated results to the \textit{zkOracle} contract while reducing the number of on-chain operations to a minimum. During the execution of this mechanism, a majority of oracle nodes query data from another blockchain and have to agree on the result. The system follows the same principle as blockchain rollups by moving computation and state storage off-chain. However, instead of rolling up transactions, the system bundles votes on queried data from another blockchain. The aggregator verifies the votes and changes the balances of committee members based on their participation during the execution of the aggregation mechanism. The aggregator only submits a \ac{zk-SNARK}, the query result, and the updated Merkle tree with the account information to the \textit{zkOracle} contract. Hence, the aggregation mechanism requires only minimal computational and storage resources from the blockchain reducing the overall costs while maintaining decentralization.

Initially, a client of the \textit{zkOracle} contract has to issue a request by calling the \textit{getBlockByNumber} function of the \textit{zkOracle} contract~(Step~1) providing the block number of a requested block and the necessary rewards for the oracle nodes executing this request. The \textit{zkOracle} contract assigns the request a unique ID and stores it. Depending on the blockchain, the oracle nodes have to query the blockchain to check for any pending requests continuously, or the blockchain emits an event to notify the oracle nodes~(Step~2).

After getting the information about a pending request, the oracle nodes use the \textit{getAggregator} function of the \textit{zkOracle} contract to determine their current role. The \textit{zkOracle} contract supports different approaches for selecting the current aggregator. An efficient strategy is to adopt a round-robin method for organizing each oracle node to serve as an aggregator for a specific time interval until switching to the next aggregator. This approach prevents races between multiple aggregators for aggregating results, which wastes resources and incurs additional costs for interacting with the \textit{zkOracle} contract since only one submission would be successful. Another possibility would be for the \textit{zkOracle} contract to randomly choose the next aggregator, making it harder for adversaries to attack specific aggregators upfront. However, this approach requires the computation of a random seed using an additional secret input for the aggregation mechanism and a timeout to account for malicious aggregators not providing an answer. In this case, the aggregator rotates after every block submission or after reaching the timeout.

The validators query the source blockchain~(Steps~3~and~4) to retrieve the block with the specified number and check if the block is final. Therefore, the validators need to check that such a block exists and that at least $n$ blocks confirm it in the case of blockchains with probabilistic finality. For blockchains with immediate finality, waiting for $n$ confirmations is unnecessary. Ignoring the finality of blocks can result in validators returning a wrong result due to possible forks in the source blockchain. If a validator determines that such a block does not exist or is not final, it responds with a zero value as the block hash. A validator creates a vote consisting of the validator’s ID, the request number, and the hash of the retrieved block. After that, a validator signs the vote using its private key and retrieves the current aggregator’s IP address from the \textit{zkOracle} contract. Hence, an aggregator can not forge an invalid vote. A validator establishes a direct channel with the aggregator and sends the vote with the according signature to the aggregator~(Step~5).

Upon receiving a vote from a validator, the aggregator verifies the vote and stores it in the aggregator's mempool. The aggregator considers a vote to be valid if the signature is valid. After receiving a majority of valid votes with the same result for a specific request, the aggregator starts the aggregation process. During this process, the aggregator has to execute the aggregation mechanism locally and construct the witness for the aggregation circuit to generate the \ac{zk-SNARK}, proving that the submitted result and applied state changes are correct~(Step~6). The \ac{zk-SNARK} ensures that an aggregator can only perform the correct aggregation algorithm based on the votes received from the validators and submit the majority answer.

\begin{algorithm}[t]
\caption{Aggregation Circuit}
\label{alg:aggregation}
\begin{algorithmic}[1]
\Function{aggregate}{$preStateRoot$, $postStateRoot$, $blockHash$, $request$, $validatorBits$, $a$, $validators$}
\For{$i < numAccounts$} \label{line:aggregation:duplicates}
    \For{$j < numAccounts$}
        \If{$i = j$}
            \State \textbf{continue}
        \EndIf
        \State $assert(validators[i].index \neq validators[j].index)$
    \EndFor
\EndFor \label{line:agg:endduplicate}
\State $verifyMerkleProof(preStateRoot, a.MerkleProof)$  \label{line:aggregation:merkle_aggregator}
\State $a.balance \gets a.balance + reward$ \label{line:aggregation:aggregator_reward}
\State $a.MerkleProof[0] \gets hash(a.index, a.publicKeyX, a.publicKeyY, a.balance)$ \label{line:aggregation:start_aggregator_root}
\State $intermediateRoot \gets computeRootFromPath(a.MerkleProof)$ \label{line:aggregation:end_aggregator_root}
\State $actualValidatorBits \gets 0$
\ForAll{$v \in validators$}
    \State $verifyMerkleProof(intermediateRoot, v.MerkleProof)$ \label{line:aggregation:merkle_validator}
    \State $msg \gets hash(v.index, request, blockHash)$ \label{line:aggregation:hash}
    \State $verifySignature(v.signature, v.publicKey, msg)$ \label{line:aggregation:signature}
    \State $assert(v.blockHash == blockHash)$ \label{line:aggregation:check_blockhash}
    \State $v.balance \gets v.balance + reward$ \label{line:aggregation:validator_reward}
    \State $actualValidatorBits \gets actualValidatorBits + 2^{v.index}$ \label{line:aggregation:flip_bit}
    \State $v.MerkleProof[0] \gets hash(v.index, v.publicKeyX, v.publicKeyY, v.balance)$ \label{line:aggregation:start_validator_root}
    \State $intermediateRoot \gets computeRootFromPath(v.MerkleProof)$ \label{line:aggregation:end_validator_root}
\EndFor
\State $assert(actualValidatorBits = validatorBits)$ \label{line:aggregation:check_bits}
\State $assert(intermediateRoot = postStateRoot)$ \label{line:aggregation:check_state}
\EndFunction
\end{algorithmic}
\end{algorithm}
 
The aggregator executes a program~(see~Alg.~\ref{alg:aggregation}) compiled into an arithmetic circuit for \ac{zk-SNARK} construction to aggregate the votes. The aggregation circuit expects the previous state root, the resulting state root, the expected block hash, the request number, and the validator bit vector as public inputs. The aggregation circuit also requires private inputs, including the account information of the aggregator and validators with the respective Merkle proofs. Additionally, the private inputs contain the signatures and returned block hashes from the validators. 

Initially, the aggregator must ensure no multiple votes from the same validator. For each request, a validator only has one vote. For that, the aggregator iterates over all submitted votes and checks that the validator index does not match another vote’s validator index~(Lines~\ref{line:aggregation:duplicates}-\ref{line:agg:endduplicate}). After that, the aggregator checks that the Merkle tree includes the account of the aggregator~(Line~\ref{line:aggregation:merkle_aggregator}) and increases its balance by the specified reward~(Line~\ref{line:aggregation:aggregator_reward}). After increasing the balance, the aggregator computes an intermediate state root reflecting the change in the aggregator's balance~(Lines~\ref{line:aggregation:start_aggregator_root}-\ref{line:aggregation:end_aggregator_root}). Subsequently, the aggregator processes the submitted votes from the validators. For each vote, the aggregator checks if the Merkle tree includes the validator’s account~(Line~\ref{line:aggregation:merkle_validator}) and hashes the vote to construct the message~(Line~\ref{line:aggregation:hash}). Then, the aggregator verifies the signature provided by the validator by using the validator's public key and the message~(Line~\ref{line:aggregation:signature}). After verifying that the validator is part of the system and provided an authenticated vote, the aggregator checks if the block hash provided by the validator matches the expected block hash provided by the majority of validators~(Line~\ref{line:aggregation:check_blockhash}).

The aggregator also adds the reward to the validator's balance~(Line~\ref{line:aggregation:validator_reward}) and updates a bit vector indicating which validators voted for the result to ensure data availability. A bit value of 1 at a specific position shows that the validator with this index contributed to the result. In contrast, a bit value of 0 shows that the respective validator did not participate. The aggregator flips the bits at the respective position by adding $2^{v.index}$ to the validatorBits~(Line~\ref{line:aggregation:flip_bit}). Otherwise, the other oracle nodes cannot synchronize the state without knowing which state changes are necessary since the aggregator only submits the updated Merkle root. Then, the aggregator computes a new intermediate root reflecting the state changes for the validator's account~(Lines~\ref{line:aggregation:start_validator_root}-\ref{line:aggregation:end_validator_root}). Finally, after repeating this process for each validator, the aggregator checks if the validatorBits match the provided validatorBits~(Line~\ref{line:aggregation:check_bits}) and if the intermediate root equals the provided post-state root~(Line~\ref{line:aggregation:check_state}).

After creating the witness and the public inputs for the aggregation circuit, the aggregator generates the \ac{zk-SNARK} to submit the aggregated result. The aggregator calls the $submitBlock$ function of the \textit{zkOracle} contract~(Step~7), providing the request number, block hash, validators bit vector, updated state root, and \ac{zk-SNARK}. Initially, the \textit{zkOracle} contract verifies that the transaction sender is the current aggregator, and the request is still pending. After that, the \textit{zkOracle} contract verifies the submitted \ac{zk-SNARK}~(Step~8), updates the state root, and stores the block hash. Finally, the \textit{zkOracle} contract notifies the client about the submission~(Step~9), and clients can query the \textit{zkOracle} contract for the block hash to use it for \ac{SPV}.

After successfully submitting an aggregated result to the \textit{zkOracle} contract, the aggregator executes a program~(see~Alg.~\ref{alg:slashing}) to process invalid votes and slash misbehaving validators~(Step~10). The aggregator considers a vote invalid if the hash of the retrieved block differs from the majority answer to a certain request of the other validators. The slashing circuit expects the same inputs as the aggregation circuit, except that the slashing circuit only expects a single validator account with the submitted signature and block hash. Further, the indices of the aggregator and the validator need to be public inputs to ensure data availability.

Initially, the aggregator verifies the inclusion of the validator's account in the Merkle tree~(Line~\ref{line:slashing:merkle_validator}). After that, the aggregator computes the message~(Line~\ref{line:slashing:hash}) and verifies the signature provided by the validator~(Line~\ref{line:slashing:signature}). In the next step, the aggregator sets the validator's balance to 0~(Line~\ref{line:slashing:validator_balance}) and computes the new state root~(Lines~\ref{line:slashing:start_validator_root}-\ref{line:slashing:end_validator_root}). The aggregator verifies the inclusion of the aggregator's account in the Merkle tree~(Line~\ref{line:slashing:merkle_aggregator}), adds the deducted balance to its account~(Line~\ref{line:slashing:slasher_balance}), and also computes the new state root~(Lines~\ref{line:slashing:start_aggregator_root}-\ref{line:slashing:end_aggregator_root}). Further, the aggregator checks that the block hash provided by the validator is different from the block hash provided by the majority of validators~(Line~\ref{line:slashing:check_blockhash}) and that the computed state root matches the provided post-state root~(Line~\ref{line:slashing:check_state}). Finally, the aggregator calls the $slash$ function of the \textit{zkOracle} contract~(Step~11) and provides the proof, the aggregator and validator indices, the request number, and the post-state root. The \textit{zkOracle} contract verifies that the request is not pending, verifies the proof, and stores the updated state root~(Step~12).

\begin{algorithm}[t]
\caption{Slashing Circuit}
\label{alg:slashing}
\begin{algorithmic}[1]
\Function{slash}{$preStateRoot$, $postStateRoot$, $blockHash$, $request$, $a$, $v$}
\State $verifyMerkleProof(preStateRoot, v.MerkleProof)$ \label{line:slashing:merkle_validator}
\State $msg \gets hash(v.index, request, blockHash)$ \label{line:slashing:hash}
\State $verifySignature(v.signature, v.publicKey, msg)$ \label{line:slashing:signature}
\State $v.balance \gets 0$ \label{line:slashing:validator_balance}
\State $v.MerkleProof[0] \gets hash(v.index, v.publicKeyX, v.publicKeyY, v.balance)$ \label{line:slashing:start_validator_root}
\State $intermediateRoot \gets computeRootFromPath(v.MerkleProof)$ \label{line:slashing:end_validator_root}
\State $verifyMerkleProof(intermediateRoot, a.MerkleProof)$ \label{line:slashing:merkle_aggregator}
\State $a.balance \gets a.balance + v.balance$ \label{line:slashing:slasher_balance}
\State $a.MerkleProof[0] \gets hash(a.index, a.publicKeyX, a.publicKeyY, a.balance)$ \label{line:slashing:start_aggregator_root}
\State $intermediateRoot \gets computeRootFromPath(a.MerkleProof)$ \label{line:slashing:end_aggregator_root}
\State $assert(blockHash \neq v.blockHash)$ \label{line:slashing:check_blockhash}
\State $assert(postStateRoot = intermediateRoot)$ \label{line:slashing:check_state}
\EndFunction
\end{algorithmic}
\end{algorithm}

%% file: 04_security.tex
\section{Security Analysis}
\label{sec:security}

This section provides a security analysis of our oracle solution. Initially, we clarify assumptions and the adversarial model. Then, we define the required security properties and show that our solution fulfills all these properties. Finally, we discuss Sybil-resistance and the adversarial bound.

\paragraph{Cryptographic Assumptions}

We assume that \acp{zk-SNARK} have perfect completeness, soundness, and zero-knowledge~\cite{groth2016size}. Further, we consider hash functions modeled as random oracles~\cite{bellare1993random} and digital signature schemes having \ac{EUF-CMA} security~\cite{goldwasser1988digital}.

\paragraph{System Assumptions}

We assume that, at any point in time, there is an honest majority of oracle nodes. Further, we assume that the communicating blockchains provide safety and liveness guarantees.

\paragraph{Adversarial Model and Security Properties}

We assume a computationally-bound static adversary~$\mathcal{A}$ that can corrupt a fixed number of oracle nodes prior to the execution of the protocol. However, depending on the committee's maximum size, $\mathcal{A}$ can only corrupt a minority of oracle nodes. The oracle possesses the following security properties:

\begin{definition}[\textit{Oracle Safety}]
An oracle is safe if, for any request made on a destination chain $bc_{dest}$ by a client $c$ with respect to a source chain  $bc_{src}$, it provides a correct answer to $bc_{dest}$. A correct answer to $bc_{dest}$ includes the requested data included and confirmed with overwhelming probability in $bc_{src}$.
\end{definition}

\begin{definition}[\textit{Oracle Liveness}]
An oracle is live if, for any request made on a destination chain $bc_{dest}$ by a client $c$ with respect to a source chain  $bc_{src}$, it eventually provides an answer to $bc_{dest}$. 
\end{definition}

\begin{theorem}
    Our zkOracle is secure, i.e., it is safe and live. 
\end{theorem}
\begin{proof}
Towards contradiction, suppose our zkOracle is not secure. For this to happen, it means that there was a point in time where (i) the oracle provided $bc_{dest}$ with an incorrect answer, or (ii) the oracle did not provide an answer to a request. 

In case (i), for the oracle to provide an incorrect answer, it means that either a majority of nodes were adversarial, i.e., they voted on an incorrect answer, or that the properties of the cryptographic primitives used by the oracle have been broken, i.e., the aggregator successfully forged an invalid proof or validator signatures, or that either $bc_{src}$ or $bc_{dest}$ are not secure. This contradicts our cryptographic and system assumptions. 
    
In case (ii), for the oracle to not provide an answer, it means that either none of the aggregators ever submits a valid answer, not enough signatures are produced by the nodes in the committee, or that either $bc_{src}$ or $bc_{dest}$ are not secure. This contradicts our system assumptions.
\end{proof}

\paragraph{Sybil-Resistance and Adversarial Bound} We note that the staking requirement for nodes joining the oracle provides sybil-resistance as well as the lower bound on the stake for the adversary to pull an attack against the oracle successfully. 

The oracle requires each oracle node $o_{i}$ joining the committee $C$ with a maximum size of $n$ nodes to deposit a stake $s_{i} \geq s_{min}$ where $s_{min}$ is the configurable minimum stake required to join $C$. The amount depends on the number of registered oracle nodes $|C|$. Hence, if $|C| < n$, the committee is not full and $s_{i} \geq s_{min}$. However, if $|C| = n$, i.e., the committee is already full, $o_{i}$ can only join $C$ by replacing another node $o_{j} \in C$ if  $s_{i} > s_{j}$. For the oracle to submit an answer, it requires a majority $t = \frac{n}{2}+1$ of oracle nodes to attest to the correctness of the transferred data and generate a valid proof. Hence, an adversary $\mathcal{A}$ needs to have more stake than the majority of registered oracle nodes with the lowest deposited stake $C_{P} = min_t(C)$ which have a cumulative stake $s_{P} = \sum_{o_{i} \in C_{P}} s_{i}$. Therefore, $s_{P}$ is a lower bound on the required adversary's stake $s_{\mathcal{A}}$, i.e., $s_{\mathcal{A}} > s_{P}$ to break the oracle's security successfully.

The requirement for oracle nodes to deposit a substantial stake not only increases their investment in the system's integrity but also elevates the economic cost of attempting to join $C$ with multiple identities. This requirement ensures nodes are incentivized to maintain their position within the committee, as losing their place means forfeiting potential future rewards. Hence, this mechanism effectively raises the barriers against Sybil attacks, safeguarding the system's decentralization and security.

%% file: 05_implementation.tex
\section{Implementation}
\label{sec:implementation}

In this section, we discuss the implementation details of the different smart contracts and the node software. Further, we provide the prototypical implementation as open-source software directly available on GitHub\footnote{\url{https://github.com/soberm/zkOracle}}.

\subsection{Smart Contracts}
\label{subsec:smart_contracts}

We implemented the smart contracts for the Ethereum smart contract and \ac{dApp} platform using Solidity. There are several reasons for using Ethereum to host the smart contracts of the proposed solution. Ethereum is currently the largest smart contract platform and, therefore, also comes with a rich ecosystem and huge community. Further, Ethereum provides the necessary smart contract capabilities and pre-compiled contracts, i.e., more efficient contracts, since they do not incur additional overhead using the \ac{EVM}. These pre-compiled contracts are part of the Ethereum client software and allow, among other things, the efficient verification of \ac{zk-SNARKs} using the \textit{altbn128} curve, which is a necessary prerequisite to implement the proposed solution. Many blockchains also use the \ac{EVM} and support the implemented smart contracts without requiring changes. These reasons render Ethereum the most viable solution for implementing the prototype. Further, porting or implementing the solution to other smart contract platforms that do not use the \ac{EVM} is also possible. However, these platforms must bring the necessary smart contract capabilities and allow efficient elliptic curve pairing checks.

The \textit{zkOracle} contract inherits an additional smart contract implementing a sparse Merkle tree to store the accounts since inserting new data is easier. The \textit{zkOracle} contract applies MiMC~\cite{albrecht2016mimc} for hashing the data, an arithmetization-friendly hash function optimized for modular arithmetic in a finite field. Therefore, the oracle nodes can efficiently compute all operations regarding the Merkle tree inside an arithmetic circuit. Using arithmetization-friendly hash functions is necessary to minimize resource consumption during proof generation. Further, the oracle randomly chooses the next aggregator by adding a random elliptic curve point as an additional public input to the aggregation circuit. The aggregator computes the next seed during the aggregation mechanism by multiplying the seed with its private key. The \textit{zkOracle} contract uses an exit time of seven days to ensure misbehaving clients can not arbitrarily remove other oracle nodes from the committee. For that, the \textit{zkOracle} contract retrieves the current block time stamp and adds a seven-day delay. For the verification of the \ac{zk-SNARKs}, we use two additional verifier contracts, which are automatically generated with the help of the \textit{gnark} zk-SNARK library~\cite{gnark-v0.8.0}, also used by the provided client software. These generated contracts use a pre-compiled contract by Ethereum to compute elliptic curve pairings to verify the submitted proofs.

\subsection{Client}
\label{subsec:client}

We implemented the client software for the oracle nodes using the Go\footnote{https://go.dev/} programming language. The client software uses automatically generated contract bindings to interact with the smart contracts to avoid boilerplate code and to enable seamless interaction. For communication, the oracle nodes use gRPC\footnote{https://grpc.io/} to establish point-to-point channels, with the current aggregator offering an RPC API for validators to submit their votes. The client software allows each oracle node to act as an aggregator and a validator. 

Each oracle node has to store and continuously synchronize the state of the \textit{zkOracle} contract. For that, an oracle node stores the raw state in memory and listens to all events which indicate state changes, i.e., changes to the accounts stored in the Merkle tree. These events include registering, replacing, and withdrawing oracle nodes and submitting results or validator slashing. The oracle node needs to execute the state changes in the correct order and to update its locally stored state. Newly joined oracle nodes synchronize their state by starting from the beginning, i.e., the deployment time of the \textit{zkOracle} contract and executing all state changes based on the persisted events. Further, oracle nodes that went offline must synchronize their state from the last known state change. Optimizations are still possible, whereby oracle nodes download the state of already synchronized oracle nodes and verify that the root matches the root of the \textit{zkOracle} contract.

The validators listen for \textit{BlockRequested} events on the target blockchain to start the validation process. These events include a request number and the corresponding block number. After receiving a \textit{BlockRequested} event, a validator uses the JSON-RPC client to connect to an Ethereum node and retrieve the block with the given block number. Then, a validator uses a predefined threshold to determine the finality of the retrieved block. Afterward, a validator creates the vote and signs it using \ac{EdDSA}. The applied \ac{EdDSA} does not use \textit{Curve25519}~\cite{bernstein2006curve25519} as the aggregator verifies the signature in an arithmetic circuit using modular arithmetic over a different field. Therefore, the validator uses \ac{EdDSA} with a twisted Edwards curve defined over the same field as the arithmetic circuit~\cite{belles2021twisted}. Further, a validator uses the MiMC hash function to create the message to be signed. Finally, the validator calls the \textit{zkOracle} contract to retrieve the current aggregator with its IP address and sends the signed vote to the current aggregator using gRPC. While this implementation only supports Ethereum-based source blockchains, it is easily possible to support other blockchains by simply switching the client of the source blockchain and defining a new strategy to determine the finality of a block.

The aggregator receives votes from the validators and stores the votes in its in-memory mempool until reaching a majority of votes with the same result. Afterward, the aggregator fetches the state and constructs the witness for the aggregation circuit created with the gnark \ac{zk-SNARKs} library. For generating the proof, the aggregator uses the Groth16 proving scheme and transforms the proof to the respective format for the verifier contract.

%% file: 06_evaluation.tex
\section{Evaluation}
\label{sec:evaluation}

After providing the implementation details, we use the prototypical implementation to evaluate the proposed solution. Initially, we conduct several experiments to measure the costs of using smart contracts. Afterward, we examine the performance and memory requirements for generating the proofs.

\subsection{Experimental Setup}
\label{subsec:experimental_setup}

We perform the experiments on a machine equipped with an Intel Core i7-10510U CPU running Ubuntu 20.04.6 using 8 cores with a clock speed of 1.80 GHz. The machine has 16GB LPDDR3 RAM with a frequency of 2133 MHz and a 2 TB SSD. We deploy the smart contracts on a local Ethereum blockchain using HardHat Network 2.11.1. After deploying the smart contracts, we execute each function of the \textit{zkOracle} contract to measure the gas costs. To measure the proof generation time and memory consumption, we simulate a committee of oracle nodes responding to requests. After every 50\textsuperscript{th} measurement, we double the committee size and repeat the experiment until the committee size is 256. The program's source code for simulating the oracle nodes and measuring the proof generation time and memory consumption is again available on GitHub.

\subsection{Costs}
\label{subsec:costs}

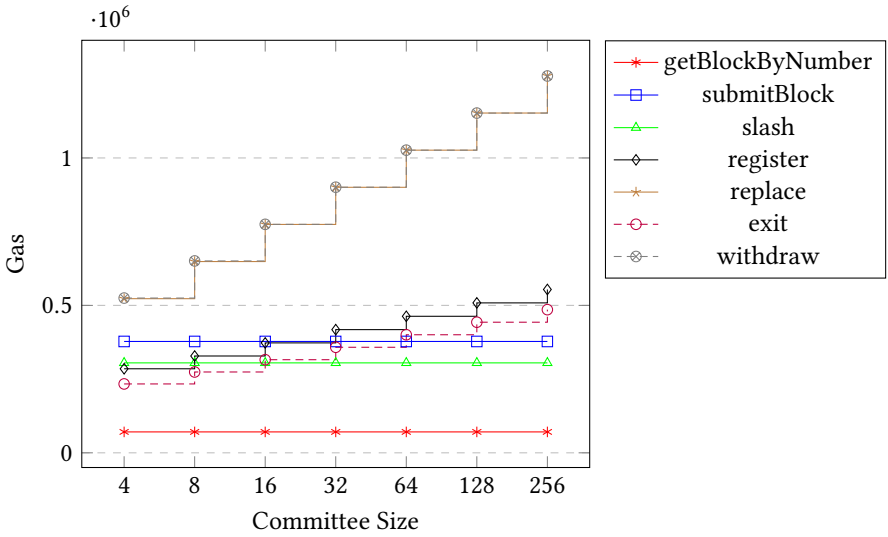
\begin{figure}[t]
\centering
\begin{adjustbox}{width=0.85\linewidth}
\begin{tikzpicture}
\begin{axis}[
    xmode=log,
    log basis x={2},
    xlabel={Committee Size},
    ylabel={Gas},
    xtick={4,8,16,32,64,128,256},
    xticklabels={4,8,16,32,64,128,256},
    legend pos=outer north east,
    ymajorgrids=true,
    grid style=dashed
]

\addplot+[
    const plot,
    color=red,
    mark=asterisk,
    ]
    coordinates {
    (4,70914)(8,70914)(16,70914)(32,70914)(64,70914)(128,70914)(256,70914)
    };
    \addlegendentry{getBlockByNumber}

\addplot+[
    const plot,
    color=blue,
    mark=square,
    ]
    coordinates {
    (4,377673)(8,377673)(16,377673)(32,377673)(64,377673)(128,377673)(256,377673)
    };
    \addlegendentry{submitBlock}

\addplot+[
    const plot,
    color=green,
    mark=triangle,
    ]
    coordinates {
    (4,305024)(8,305024)(16,305024)(32,305024)(64,305024)(128,305024)(256,305024)
    };
    \addlegendentry{slash}

\addplot+[
    const plot,
    color=black,
    mark=diamond,
    ]
    coordinates {
    (4,285008)(8,328365)(16,372800)(32,417783)(64,463047)(128,508454)(256,553940)
    };
    \addlegendentry{register}

\addplot+[
    const plot,
    color=brown,
    mark=star,
    ]
    coordinates {
    (4,522704)(8,648610)(16,774565)(32,900485)(64,1026418)(128,1152327)(256,1278273)
    };
    \addlegendentry{replace}

\addplot+[
    const plot,
    color=purple,
    mark=o,
    ]
    coordinates {
    (4,233602)(8,274146)(16,315784)(32,357952)(64,400392)(128,442955)(256,485572)
    };
    \addlegendentry{exit}

\addplot+[
    const plot,
    color=gray,
    mark=otimes,
    ]
    coordinates {
    (4,524883)(8,650789)(16,774908)(32,900815)(64,1026433)(128,1152192)(256,1277974)
    };
    \addlegendentry{withdraw}
    
\end{axis}
\end{tikzpicture}
\end{adjustbox}
\caption{Gas consumption of the \textit{zkOracle} contract's functions}
\label{fig:gas}
\end{figure}

The Ethereum network uses gas as a unit of measurement to quantify the computational effort and storage space for executing transactions. Every user issuing a transaction has to pay transaction fees using \ac{ETH} based on the amount of gas spent to execute the transaction. This mechanism serves as the basis for incentivizing block producers and validators to spend resources on operating and securing the network.  It also helps to manage and protect the network's resources from misuse. The costs are an essential aspect to pay attention to when designing and developing decentralized applications since the costs determine a considerable part of the practical applicability of a solution since users have to pay for every interaction and reason whether the benefits outweigh the costs. Therefore, we examine the gas consumption of the \textit{register}, \textit{replace}, \textit{exit}, \textit{withdraw}, \textit{getBlockByNumber}, \textit{submitBlock}, and \textit{slash} functions implemented by the \textit{zkOracle} contract.

The results of our experiments~(see~Fig.~\ref{fig:gas}) show that the \textit{getBlockByNumber}, \textit{submitBlock}, and \textit{slash} functions have a constant gas consumption. The \textit{getBlockByNumber} function consumes 70,914~gas, the \textit{submitBlock} function 377,673~gas, and the \textit{slash} function 305,024~gas. The \textit{submitBlock} function has low and constant gas consumption since its only task is to store the request and notify the oracle nodes. The \textit{submitBlock} and \textit{slash} functions consume more gas since these functions need to verify \ac{zk-SNARKs}, including the computation of elliptic curve pairings. However, due to the application of \ac{zk-SNARKs}, the gas consumption for verifying the correct execution of the off-chain aggregation and slashing mechanisms is low and independent of the committee size.

The remaining functions of the \textit{zkOracle} contract to manage the committee of oracle nodes exhibit an increasing gas consumption depending on the height of the Merkle tree storing the accounts. The gas consumption for these functions remains nearly constant, with negligible variations due to the \ac{EVM}'s memory and storage management, until the height of the Merkle tree increases. The \textit{register} function consumes 285,008 gas, the \textit{replace} function 522,704 gas, the \textit{exit} function 233,602, and the \textit{withdraw} function 524,883 gas considering a committee size of 4. Increasing the committee size to 256 leads to a gas consumption of 553,940 for the \textit{register} function, 1,278,273 for the \textit{replace} function, 485,572 for the \textit{exit} function, and 1,277,974 for the \textit{withdraw} function. The \textit{replace} and \textit{withdraw} functions have nearly the same gas consumption, as withdrawing is a replacement with an empty account and an additional transfer of funds to the account owner. The gas consumption of these operations can get quite high since the number of expensive hash operations increases. However, these operations are rarely executed and only needed to change the committee of oracle nodes.

Compared to another oracle solution, e.g., the oracle in~\cite{sober2021voting}, which consumes 257,602 gas $\pm$ 21,671 gas for block submissions, the proposed solution consumes more gas. However, it does not require executing a \ac{DKG} protocol and enables equal reward distribution. Further, the proposed solution incurs fewer costs than a relay solution, e.g., ETH relay~\cite{frauenthaler2020eth}, which consumes 284,041~gas for each block header submission. The submission costs of ETH relay may be lower, but the relay requires off-chain clients to relay each block of the source blockchain, while the proposed oracle solution only retrieves blocks on-demand. The \ac{OCR} protocol~\cite{breidenbach2021ocr} used by Chainlink's price feed oracles consumes around 291 kGas for the first submission using 31 oracle nodes. However, these costs increase with the number of oracle nodes while our solution provides constant submission costs. Another comparison with zkBridge~\cite{xie2022zkbridge} shows that while zkBridge only consumes 220 kGas for on-chain verification, it requires tailored solutions based on the connected blockchains, while our solution is mostly blockchain agnostic. Further, there is still the potential to reduce gas consumption by changing the public inputs to private inputs and only providing a hash of the previously public inputs to the circuits to check if the inputs are correct. However, this comes at the cost of higher resource consumption for generating the proofs since most blockchains do not support arithmetization-friendly hash functions.

\subsection{Performance}
\label{subsec:performance}

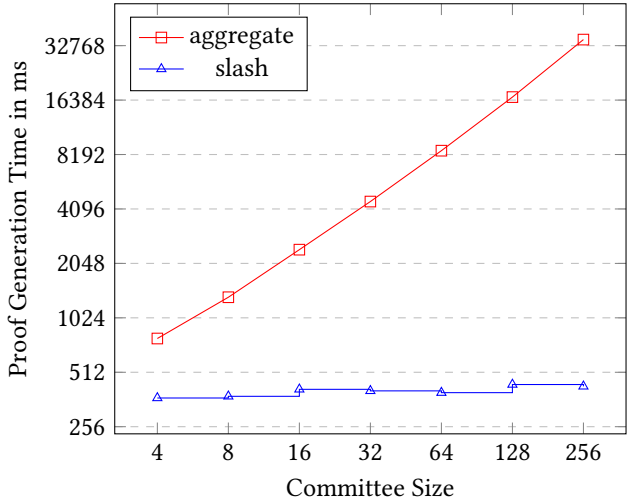
\begin{figure}[t]
\centering
\begin{adjustbox}{width=0.6\linewidth}
\begin{tikzpicture}
\begin{axis}[
    xmode=log,
    ymode=log,
    log basis x={2},
    log basis y={2},
    xlabel={Committee Size},
    ylabel={Proof Generation Time in ms},
    xtick={4,8,16,32,64,128,256,1024,2048,4096,8192,16384,32768},
    xticklabels={4,8,16,32,64,128,256,1024,2048,4096,8192,16384,32768},
    ytick={4,8,16,32,64,128,256,512,1024,2048,4096,8192,16384,32768},
    yticklabels={4,8,16,32,64,128,256,512,1024,2048,4096,8192,16384,32768},
    legend pos=north west,
    ymajorgrids=true,
    grid style=dashed
]

\addplot[
    color=red,
    mark=square,
    ]
    coordinates {
    (4,787)(8,1331)(16,2439)(32,4501)(64,8602)(128,17021)(256,35397)
    };
    \addlegendentry{aggregate}

\addplot+[
    const plot,
    color=blue,
    mark=triangle,
    ]
    coordinates {
    (4,369)(8,377)(16,412)(32,404)(64,395)(128,438)(256,428)
    };
    \addlegendentry{slash}
    
\end{axis}
\end{tikzpicture}
\end{adjustbox}
\caption{Time to generate the proofs}
\label{fig:time}
\end{figure}

As has already been discussed before, the execution of smart contracts introduces additional costs for clients and oracle nodes. Therefore, the oracle nodes move as much computation and storage off-chain to reduce the costs. The aggregator executes the aggregation and slashing mechanisms locally and generates \ac{zk-SNARKs} to allow the \textit{zkOracle} contract to verify these computations on the blockchain. However, in exchange, the current aggregator must spend a lot of computational resources to generate a \ac{zk-SNARK}. The aggregator's time to generate the proofs is a critical aspect affecting the system's throughput. The proof generation time depends mainly on the complexity of the computations, i.e., the aggregation and slashing algorithms. The more complex the calculation, the larger the arithmetic circuit and the longer it takes to generate the proof. Therefore, we examine the influence of the committee size on the proof generation time.

The results~(see~Fig.~\ref{fig:time}) show that the proof generation time for the aggregation algorithm increases linearly with the committee size. The proof generation time for the slashing algorithm, on the other hand, is constant in intervals depending on the height of the Merkle tree. However, the proof generation time for the aggregation algorithm increases considerably faster since a bigger committee size requires more votes from the validators and, therefore, more Merkle proofs, signatures, and value transfers, leading to a bigger circuit. Generating the proof for the aggregation algorithm takes 787 ms $\pm$ 14 ms with a committee size of 4 up to 35.4 s $\pm$ 137 ms with a committee size of 256. The proof generation time for the slashing algorithm is 369 ms $\pm$ 12 ms for a committee of size 4 and 428 ms $\pm$ 2 ms for a committee of size 256. The measured proof generation time allows for the practical application of the proposed solution since blockchains exhibit block times in the range of multiple seconds, limiting the overall throughput.

\subsection{Memory}
\label{subsec:memory}

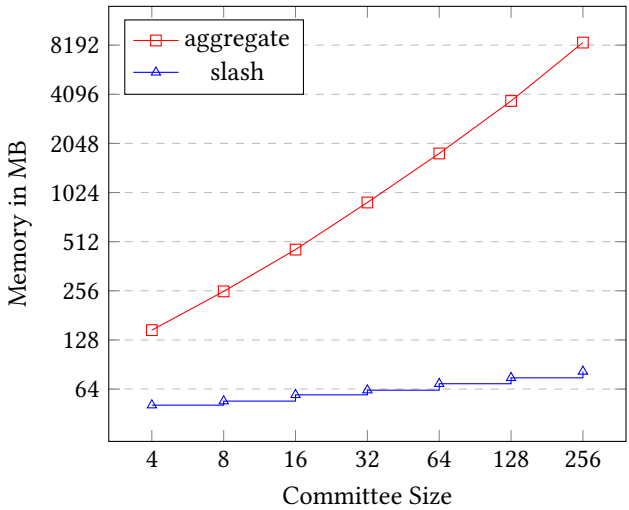
\begin{figure}[t]
\centering
\begin{adjustbox}{width=0.6\linewidth}
\begin{tikzpicture}
\begin{axis}[
    xmode=log,
    ymode=log,
    log basis x={2},
    log basis y={2},
    xlabel={Committee Size},
    ylabel={Memory in MB},
    xtick={4,8,16,32,64,128,256},
    xticklabels={4,8,16,32,64,128,256},
    ytick={64,128,256,512,1024,2048,4096,8192,16384,32768},
    yticklabels={64,128,256,512,1024,2048,4096,8192,16384,32768},
    legend pos=north west,
    ymajorgrids=true,
    grid style=dashed
]

\addplot[
    color=red,
    mark=square,
    ]
    coordinates {
    (4,147)(8,254)(16,458)(32,890)(64,1777)(128,3727)(256,8493)
    };
    \addlegendentry{aggregate}

\addplot+[
    const plot,
    color=blue,
    mark=triangle,
    ]
    coordinates {
    (4,51)(8,54)(16,59)(32,63)(64,69)(128,75)(256,82)
    };
    \addlegendentry{slash}
    
\end{axis}
\end{tikzpicture}
\end{adjustbox}
\caption{Memory usage during proof generation}
\label{fig:memory}
\end{figure}

Memory consumption is another critical aspect to consider for the applicability of the proposed solution in a real-world scenario. The reason for this is a higher entry barrier for oracle nodes and, consequently, a possibly lower degree of decentralization, as oracle nodes require more memory to execute the aggregation mechanism. Therefore, after examining the proof generation time, we also inspect the memory consumption. As with the proof generation time, the memory consumption also depends on the complexity of the algorithms due to the corresponding larger size of the arithmetic circuit. Therefore, we also inspect the memory consumption during proof generation for a varying committee size. For that, we consider the allocated memory for building the constraint system, the proving key, creating the witness, and generating the proof.

The results~(see~Fig.~\ref{fig:memory}) show that memory consumption follows the same behavior as the proof generation time. The memory consumption increases linearly for the aggregation algorithm and stepwise for the slashing algorithm. As with the proving time, the increasing complexity of the aggregation circuit makes the memory consumption higher and faster growing than for the slashing circuit. The generation of the aggregation proof requires 147 MB with a committee size of 4 and 8.5 GB with a committee size of 256. The generation of the slashing proof, on the other hand, only requires 51 MB for a committee size of 4 and 82 MB for a committee size of 256. Therefore, oracle nodes, which do not have extensive resources, can still join the system and support fairly large committee sizes.

%% file: 07_related_work.tex
\section{Related Work}
\label{sec:related_work}

In the course of this section, we discuss related work concerning cross-blockchain communication mechanisms and protocols. Among other things, we examine blockchain relays, notary schemes, blockchains of blockchains, and others.

\subsection{Blockchain Relays}
\label{sec:sidechains_relays}

The authors of~\cite{frauenthaler2020eth} propose ETH relay, a relay scheme for Ethereum-based blockchains which use the \ac{PoW} consensus mechanism. ETH relay applies a validation-on-demand pattern where the relay contract optimistically accepts new block headers from the source blockchain. Off-chain clients continuously relay block headers from the source blockchain to the target blockchain but also monitor the relay contract to detect invalid block headers. Any off-chain client can trigger the validation of a submitted block header during a dispute period. If the block is invalid, the relay contract removes the invalid and all subsequent block headers. The relay considers a block valid if there is no dispute and the dispute period has ended. Additionally, the relay applies an incentive mechanism to encourage honest behavior and reward off-chain clients for block header submissions. While the solution achieves considerable improvements concerning cost efficiency, the cost overhead of relaying all block headers remains high.

Scaffino et al.~\cite{scaffino2023glimpse} propose a cross-blockchain communication mechanism that retrieves on-demand state information about a \ac{PoW} source blockchain on a target blockchain. A contract on the target blockchain can verify constant-sized proofs consisting of $n$ confirmation block headers showing that certain transactions are part of the source blockchain. The proof construction is similar to stateless \ac{SPV} but uses additional mechanisms to prevent upfront mining attacks and does not require Turing-complete smart contracts. Further, the authors prove the protocol’s security and show how it enables several cross-chain applications such as lending, Proofs-of-Burn, and off-chain betting hubs.

Westerkamp and Eberhardt~\cite{westerkamp2020zkrelay} propose zkRelay, another relay solution for \ac{PoW} blockchains. The relay uses \ac{zk-SNARKs} to validate block headers off-chain to reduce the costs of block header submissions. Off-chain clients generate a \ac{zk-SNARK} for an off-chain validation program which takes a block header as input and returns a boolean value indicating if the block is valid. Further, zkRelay supports batch verification of block headers, while the relay contract only needs to store the last block of a batch. The validation program builds a Merkle tree for the intermediary block headers and returns the root along with the validation result. After that, the relay allows the submission of intermediary block headers, which only requires a Merkle proof to show the inclusion of the block. The application of \ac{zk-SNARKs} to verify the blocks off-chain leads to better cost-efficiency. However, off-chain clients need more computing power and memory to generate the proofs.

Another protocol for cross-blockchain communication based on \ac{zk-SNARKs} is Zendoo~\cite{garoffolo2020zendoo}. Zendoo allows sidechain constructions for Bitcoin-like blockchains where sidechains directly monitor their respective mainchain. However, the mainchain only receives \ac{zk-SNARKs} to verify the state of the sidechain. Sidechains can have different internal structures and follow distinct approaches to create certificates for communicating with the mainchain as long as they follow the basic verifier interface. Unfortunately, Zendoo only allows the transfer of assets and not arbitrary data.

Similar to~\cite{garoffolo2020zendoo} and ~\cite{westerkamp2020zkrelay}, the authors of~\cite{xie2022zkbridge} propose a trustless blockchain bridge that uses \ac{zk-SNARKs} to reduce on-chain costs. The generation of~\ac{zk-SNARKs} requires expressing the validation rules of the connected blockchains in an arithmetic circuit which can be too inefficient for certain operations. Therefore, the authors propose a 2-layer recursive proof system. The first layer applies the proposed proof system \textit{deVirgo}, which allows for the parallel generation of the proof to achieve better performance. The second layer uses Groth16 for recursively verifying proofs of the first layer to reduce the on-chain verification costs of the proof. The smart contract checks if the block is included and verifies the submitted proof against the current state and given block header. Further, zkBridge also supports batch submissions of block headers.

Westerkamp and Diez~\cite{westerkamp2022verilay} propose Verilay, a blockchain relay for \ac{PoS} blockchains. The relay contract on the target blockchain starts from a trusted state containing the starting block and the corresponding validator set. After that, any user can update the relay state by retrieving a specific block, the current and next validator set, and submitting it to the relay contract on the target blockchain. The relay contract validates a block by checking the signatures and replaces the currently stored block if the signatures are valid. Further, the relay checks whether to transition to a new validator set. Verilay only requires a single update for each validator set period since every block contains the history of the whole chain. The authors implement Verilay for the Ethereum~2.0 beacon chain and show that Verilay is more cost-efficient than \ac{PoW} relays. However, the relay only ensures accountable safety during a trusting period when the validators can not withdraw the bonded stake. Otherwise, validators can misbehave without consequences. Further, the complexity of validator set changes can pose a challenge if it is not efficiently verifiable by the relay contract.

The discussed relay schemes enable cross-blockchain communication with a high degree of decentralization. However, these solutions heavily depend on the applied consensus mechanism of the source blockchain, as a smart contract or arithmetic circuit must be able to express the source blockchain's consensus mechanism. As a result, the solutions have a high complexity level and require many resources. On the other hand, the proposed oracle solution requires little to no changes depending on the connected blockchains. Therefore, the complexity is also correspondingly lower, leading to considerable resource savings. However, the degree of decentralization is lower as the oracle relies on a fixed-size committee.

\subsection{Notary Schemes}
\label{sec:notary_schemes}

In~\cite{sober2021voting}, the authors propose a voting-based oracle solution for cross-blockchain communication. The oracle applies an off-chain voting mechanism based on threshold signatures to allow a committee to collaboratively return information queried from another blockchain while keeping on-chain computation at a bare minimum. For that, the oracle divides its participants into an aggregator and multiple validators. The validators query the source blockchain for the requested data and sign the result with their private key shares. The aggregator collects results and the corresponding signature shares from the validators until the aggregator can create the threshold signature. After that, the aggregator returns the data and the threshold signature to an oracle contract which only has to verify the threshold signature to verify the authenticity and integrity of the data. Further, the oracle applies a crypto-economic incentive mechanism to give aggregators a fixed reward and the chance to win an additional reward, increasing over time until an aggregator wins.

The decentralized oracle network Chainlink~\cite{breidenbach2021chainlink} has recently introduced the \ac{CCIP} to allow cross-chain messages and token transfers. For that, Chainlink extends the application of its \ac{OCR} protocol~\cite{breidenbach2021ocr} to the domain of cross-blockchain communication. The \ac{OCR} protocol allows multiple oracle nodes to aggregate their results in a single report and only to submit an aggregated transaction. The protocol's main purpose is to retrieve data from external systems. However, it also allows for efficient and secure off-chain computations enabling cross-chain messaging capabilities. A smart contract can invoke a Chainlink message router on the source blockchain, which uses the decentralized oracle network to transport the message to a message router on the target blockchain. The message router on the target blockchain verifies the message and forwards the message to another smart contract.

The authors of~\cite{abdullah2022chain} propose Chain-net, a blockchain interoperability framework for cross-blockchain communication. The solution uses additional gateway nodes for each blockchain network. The gateway nodes establish the connection between different blockchains and forward cross-chain transactions. Further, the gateway ensures the atomicity of cross-chain transactions by waiting for the necessary transaction confirmations of the involved blockchains. The authors present a possible application in the healthcare domain and evaluate it against certain criteria. Unfortunately, the gateway nodes constitute a central entity in the system, contrary to blockchains' decentralized nature.

Madine et al.~\cite{madine2021appxchain} propose appXchain, a blockchain interoperability solution based on a \ac{dApp} that manages the execution of cross-blockchain transactions. The \ac{dApp} communicates with the default APIs of the underlying blockchain networks and acts as a translation layer between blockchain networks. The system requires verifier nodes for each blockchain to retrieve and verify the state from other blockchains and applies a reputation system to encourage honest behavior. Further, the authors explain the proposed solution using an example in healthcare to manage patient data.

Robinson and Ramesh~\cite{gpact2021robinson} propose a general-purpose atomic cross-chain transaction protocol. The protocol allows synchronous cross-blockchain smart contract calls for Ethereum-based blockchains. For the execution of a function call involving multiple blockchains, the application needs to simulate the execution of a call graph on all chains to determine the respective input parameters to the function calls. After that, a cross-chain control contract manages the execution on each blockchain. Further, each blockchain requires a registrar contract to determine a set of signers that attest to the emitted event's correctness during function calls.

In~\cite{nissl2021crossblockchain}, the authors propose a framework for asynchronous cross-blockchain smart contract calls. The framework requires intermediaries which broker information between the source blockchain and the target blockchain, i.e., relay specific smart contract calls and return the results. For that, the intermediaries offer to execute cross-blockchain smart contracts calls for a reward. The framework applies crypto-economic incentives and relies on a trusted set of validators that vote on the validity of smart contract invocations. More specifically, the validators check that intermediaries execute the correct function, return the correct results, and use the right amount of computation steps.

The above-mentioned solutions allow for cross-blockchain smart contract calls and general cross-blockchain communication. Unfortunately, some of these solutions require many on-chain resources, are centralized, or need more details on the design and implementation. The approach presented in~\cite{sober2021voting} comes closest to the work at hand, whereby our proposed solution does not require a \ac{DKG} protocol or a lottery-like mechanism for reward distribution. Our solution does not require any additional protocols, probably jeopardizing security. Further, the proposed oracle allows splitting the reward equally between the participating entities without incurring additional costs.

\subsection{Blockchain of Blockchains}

The Cosmos~\cite{kwon2020cosmos} network connects multiple independent blockchains, also called zones, which can seamlessly interoperate with each other. Cosmos uses an \ac{IBC}~\cite{goes2020interblockchain} to transfer arbitrary data between blockchains. The \ac{IBC} protocol follows the basic concepts of TCP/IP. \ac{IBC} allows modules on different blockchains to establish channels to transmit packets containing arbitrary data, e.g., information about a token transfer. The connected blockchains verify the transferred data packets by checking the inclusion of the data packet in the respective blockchain state. Further, off-chain clients with access to both blockchains handle the actual transmission of the packets between the ledgers.

Another blockchain network aiming for blockchain interoperability is Polkadot~\cite{wood2016polkadot}. Polkadot connects multiple so-called parachains through a relay chain, providing shared security and the possibility of trustless communication between all connected parachains. The main actors of the system are nominators, validators, and collators. The validators are randomly divided into subsets for each parachain and add new blocks to the relay chain. Polkadot uses the \ac{XCM} for cross-blockchain communication, which defines a communication language for parachains. During cross-blockchain communication, nodes place cross-chain messages in message queues and continuously ping other nodes to retrieve new messages. When assembling a new block, nodes process cross-chain messages and transactions.

In~\cite{liu2019hyperservice}, the authors propose HyperService, a platform for blockchain interoperability that allows the execution of cross-blockchain smart contracts. The platform introduces a new programming language to develop cross-blockchain applications under a universal state model which abstracts state transitions across multiple blockchains. HyperService uses a universal inter-blockchain protocol to ensure the atomic and trustless execution of the transactions submitted to the underlying blockchains. A blockchain of blockchains provides a unified view of cross-blockchain applications by recording the transaction’s status. At the same time, insurance smart contracts determine the correctness of the execution and revert transactions in case of exceptions.

Yeh et al.~\cite{yeh2022secure} propose a cross-blockchain communication mechanism based on a relay chain and certificateless signatures. The relay chain connects two heterogeneous blockchains and applies a \ac{DPoS} consensus mechanism independent of the consensus mechanism of the connected blockchains. During the execution of a cross-blockchain call, one party on each of the connected blockchains and a block producer collaborate to execute a cross-blockchain transaction and create the respective block in the relay chain. The authors implement the proposed solution for two private Ethereum-based blockchains and evaluate the solution's performance and security.

The systems based on the blockchain of blockchains approach enable cross-blockchain communication through a central relay chain which records cross-blockchain transactions and attests to the correctness of the transferred data between the connected blockchains. However, this approach also entails considerable overhead, requiring an additional blockchain network to manage the relay chain with its consensus rules. Therefore, the proposed oracle solution represents a far more lightweight solution since the oracle builds upon already existing blockchains.

%% file: 08_conclusion.tex
\section{Conclusion}
\label{sec:conclusion}

Cross-blockchain communication is one of the most pressing problems in enabling blockchain interoperability and connecting heterogenous blockchains to form a unified landscape of blockchain systems. Research has led to different solutions for achieving blockchain interoperability by proposing different cross-blockchain communication mechanisms. However, these solutions are primarily complex based on the connected blockchains, require additional blockchains, or are very resource intensive.

Therefore, we propose an off-chain communication mechanism based on oracles. The oracle comprises a committee of oracle nodes that query the source blockchain for specific data. The \textit{zkOracle} contract only stores the root of a Merkle tree, storing the account data of the oracle nodes to move state and computation off-chain. For each request, the oracle nodes need to execute an off-chain aggregation mechanism to aggregate the retrieved data from the oracle nodes and reach a consensus on the returned result. The oracle schedules one oracle node as an aggregator while the remaining nodes act as validators. The validators query the source blockchain to retrieve the data and send the signed result to the aggregator. Then, the aggregator verifies that the majority returned the same data and distributes the rewards. For this computation, the aggregator generates a \ac{zk-SNARK} and only submits the new state root, result, and proof to the \textit{zkOracle} contract. The contract must only verify the proof's correctness and store the updated state root and result. Additionally, the aggregator can slash misbehaving validators to ensure accountability and incentivize honest behavior. The proposed oracle solution requires minimal computation and storage on the blockchain and offers the necessary functionality to enable the decentralized transfer of arbitrary data between blockchains.

In future work, we want to explore possible optimizations through batch submissions to reduce the costs of the aggregator for submitting an aggregated result. Further, we want to explore possibilities to move the remaining state of the \textit{zkOracle} contract off-chain.

%% file: 09_acks.tex
\begin{acks}
The financial support from the Austrian Federal Ministry for Digital and Economic Affairs, the National Foundation for Research, Technology, and Development, and the Christian Doppler Research Association is gratefully acknowledged.
\end{acks}

%% file: main.bbl

\begin{thebibliography}{72}


\ifx \showCODEN    \undefined \def \showCODEN     #1{\unskip}     \fi
\ifx \showDOI      \undefined \def \showDOI       #1{#1}\fi
\ifx \showISBNx    \undefined \def \showISBNx     #1{\unskip}     \fi
\ifx \showISBNxiii \undefined \def \showISBNxiii  #1{\unskip}     \fi
\ifx \showISSN     \undefined \def \showISSN      #1{\unskip}     \fi
\ifx \showLCCN     \undefined \def \showLCCN      #1{\unskip}     \fi
\ifx \shownote     \undefined \def \shownote      #1{#1}          \fi
\ifx \showarticletitle \undefined \def \showarticletitle #1{#1}   \fi
\ifx \showURL      \undefined \def \showURL       {\relax}        \fi
\providecommand\bibfield[2]{#2}
\providecommand\bibinfo[2]{#2}
\providecommand\natexlab[1]{#1}
\providecommand\showeprint[2][]{arXiv:#2}

\bibitem[Abdullah et~al\mbox{.}(2022)]%
        {abdullah2022chain}
\bibfield{author}{\bibinfo{person}{Sidrah Abdullah}, \bibinfo{person}{Junaid
  Arshad}, {and} \bibinfo{person}{Muhammad Alsadi}.}
  \bibinfo{year}{2022}\natexlab{}.
\newblock \showarticletitle{Chain-Net: An Internet-inspired Framework for
  Interoperable Blockchains}.
\newblock \bibinfo{journal}{\emph{Distributed Ledger Technologies: Research and
  Practice}} \bibinfo{volume}{1}, \bibinfo{number}{2} (\bibinfo{year}{2022}),
  \bibinfo{pages}{1--20}.
\newblock


\bibitem[Agbo et~al\mbox{.}(2019)]%
        {agbo2019blockchain}
\bibfield{author}{\bibinfo{person}{Cornelius~C Agbo}, \bibinfo{person}{Qusay~H
  Mahmoud}, {and} \bibinfo{person}{J~Mikael Eklund}.}
  \bibinfo{year}{2019}\natexlab{}.
\newblock \showarticletitle{Blockchain technology in healthcare: a systematic
  review}. In \bibinfo{booktitle}{\emph{Healthcare}}, Vol.~\bibinfo{volume}{7}.
  \bibinfo{pages}{56}.
\newblock


\bibitem[Al-Bassam et~al\mbox{.}(2018)]%
        {al2018fraud}
\bibfield{author}{\bibinfo{person}{Mustafa Al-Bassam}, \bibinfo{person}{Alberto
  Sonnino}, {and} \bibinfo{person}{Vitalik Buterin}.}
  \bibinfo{year}{2018}\natexlab{}.
\newblock \showarticletitle{Fraud proofs: Maximising light client security and
  scaling blockchains with dishonest majorities}.
\newblock \bibinfo{journal}{\emph{arXiv preprint arXiv:1809.09044}}
  (\bibinfo{year}{2018}).
\newblock


\bibitem[Al-Breiki et~al\mbox{.}(2020)]%
        {al2020trustworthy}
\bibfield{author}{\bibinfo{person}{Hamda Al-Breiki}, \bibinfo{person}{Muhammad
  Habib~Ur Rehman}, \bibinfo{person}{Khaled Salah}, {and}
  \bibinfo{person}{Davor Svetinovic}.} \bibinfo{year}{2020}\natexlab{}.
\newblock \showarticletitle{Trustworthy blockchain oracles: review, comparison,
  and open research challenges}.
\newblock \bibinfo{journal}{\emph{IEEE Access}}  \bibinfo{volume}{8}
  (\bibinfo{year}{2020}), \bibinfo{pages}{85675--85685}.
\newblock


\bibitem[Albrecht et~al\mbox{.}(2016)]%
        {albrecht2016mimc}
\bibfield{author}{\bibinfo{person}{Martin Albrecht}, \bibinfo{person}{Lorenzo
  Grassi}, \bibinfo{person}{Christian Rechberger}, \bibinfo{person}{Arnab Roy},
  {and} \bibinfo{person}{Tyge Tiessen}.} \bibinfo{year}{2016}\natexlab{}.
\newblock \showarticletitle{MiMC: Efficient encryption and cryptographic
  hashing with minimal multiplicative complexity}. In
  \bibinfo{booktitle}{\emph{International Conference on the Theory and
  Application of Cryptology and Information Security}}.
  \bibinfo{publisher}{Springer}, \bibinfo{pages}{191--219}.
\newblock


\bibitem[Almashaqbeh and Solomon(2022)]%
        {almashaqbeh2022sok}
\bibfield{author}{\bibinfo{person}{Ghada Almashaqbeh} {and}
  \bibinfo{person}{Ravital Solomon}.} \bibinfo{year}{2022}\natexlab{}.
\newblock \showarticletitle{SoK: privacy-preserving computing in the blockchain
  era}. In \bibinfo{booktitle}{\emph{2022 IEEE 7th European Symposium on
  Security and Privacy (EuroS\&P)}}. IEEE, \bibinfo{pages}{124--139}.
\newblock


\bibitem[Belchior et~al\mbox{.}(2023)]%
        {belchior2023you}
\bibfield{author}{\bibinfo{person}{Rafael Belchior}, \bibinfo{person}{Luke
  Riley}, \bibinfo{person}{Thomas Hardjono}, \bibinfo{person}{Andr{\'e}
  Vasconcelos}, {and} \bibinfo{person}{Miguel Correia}.}
  \bibinfo{year}{2023}\natexlab{}.
\newblock \showarticletitle{Do you need a distributed ledger technology
  interoperability solution?}
\newblock \bibinfo{journal}{\emph{Distributed Ledger Technologies: Research and
  Practice}} \bibinfo{volume}{2}, \bibinfo{number}{1} (\bibinfo{year}{2023}),
  \bibinfo{pages}{1--37}.
\newblock


\bibitem[Belchior et~al\mbox{.}(2021)]%
        {belchior2021survey}
\bibfield{author}{\bibinfo{person}{Rafael Belchior}, \bibinfo{person}{Andr{\'e}
  Vasconcelos}, \bibinfo{person}{S{\'e}rgio Guerreiro}, {and}
  \bibinfo{person}{Miguel Correia}.} \bibinfo{year}{2021}\natexlab{}.
\newblock \showarticletitle{A survey on blockchain interoperability: Past,
  present, and future trends}.
\newblock \bibinfo{journal}{\emph{Comput. Surveys}} \bibinfo{volume}{54},
  \bibinfo{number}{8} (\bibinfo{year}{2021}), \bibinfo{pages}{1--41}.
\newblock


\bibitem[Bellare and Rogaway(1993)]%
        {bellare1993random}
\bibfield{author}{\bibinfo{person}{Mihir Bellare} {and}
  \bibinfo{person}{Phillip Rogaway}.} \bibinfo{year}{1993}\natexlab{}.
\newblock \showarticletitle{Random oracles are practical: A paradigm for
  designing efficient protocols}. In \bibinfo{booktitle}{\emph{Proceedings of
  the 1st ACM Conference on Computer and Communications Security}}.
  \bibinfo{pages}{62--73}.
\newblock


\bibitem[Bell{\'e}s-Mu{\~n}oz et~al\mbox{.}(2021)]%
        {belles2021twisted}
\bibfield{author}{\bibinfo{person}{Marta Bell{\'e}s-Mu{\~n}oz},
  \bibinfo{person}{Barry Whitehat}, \bibinfo{person}{Jordi Baylina},
  \bibinfo{person}{Vanesa Daza}, {and} \bibinfo{person}{Jose~Luis
  Mu{\~n}oz-Tapia}.} \bibinfo{year}{2021}\natexlab{}.
\newblock \showarticletitle{Twisted Edwards elliptic curves for zero-knowledge
  circuits}.
\newblock \bibinfo{journal}{\emph{Mathematics}} \bibinfo{volume}{9},
  \bibinfo{number}{23} (\bibinfo{year}{2021}).
\newblock


\bibitem[Ben-Sasson et~al\mbox{.}(2019)]%
        {ben2019scalable}
\bibfield{author}{\bibinfo{person}{Eli Ben-Sasson}, \bibinfo{person}{Iddo
  Bentov}, \bibinfo{person}{Yinon Horesh}, {and} \bibinfo{person}{Michael
  Riabzev}.} \bibinfo{year}{2019}\natexlab{}.
\newblock \showarticletitle{Scalable zero knowledge with no trusted setup}. In
  \bibinfo{booktitle}{\emph{39th Annual International Cryptology Conference}}.
  \bibinfo{publisher}{Springer}, \bibinfo{pages}{701--732}.
\newblock


\bibitem[Ben-Sasson et~al\mbox{.}(2014a)]%
        {sasson2014zerocash}
\bibfield{author}{\bibinfo{person}{Eli Ben-Sasson}, \bibinfo{person}{Alessandro
  Chiesa}, \bibinfo{person}{Christina Garman}, \bibinfo{person}{Matthew Green},
  \bibinfo{person}{Ian Miers}, \bibinfo{person}{Eran Tromer}, {and}
  \bibinfo{person}{Madars Virza}.} \bibinfo{year}{2014}\natexlab{a}.
\newblock \showarticletitle{Zerocash: Decentralized Anonymous Payments from
  Bitcoin}. In \bibinfo{booktitle}{\emph{2014 IEEE Symposium on Security and
  Privacy}}. \bibinfo{pages}{459--474}.
\newblock


\bibitem[Ben-Sasson et~al\mbox{.}(2014b)]%
        {ben2014succinct}
\bibfield{author}{\bibinfo{person}{Eli Ben-Sasson}, \bibinfo{person}{Alessandro
  Chiesa}, \bibinfo{person}{Eran Tromer}, {and} \bibinfo{person}{Madars
  Virza}.} \bibinfo{year}{2014}\natexlab{b}.
\newblock \showarticletitle{Succinct $\{$Non-Interactive$\}$ zero knowledge for
  a von neumann architecture}. In \bibinfo{booktitle}{\emph{23rd USENIX
  Security Symposium}}. \bibinfo{publisher}{USENIX Association},
  \bibinfo{pages}{781--796}.
\newblock


\bibitem[Bernstein(2006)]%
        {bernstein2006curve25519}
\bibfield{author}{\bibinfo{person}{Daniel~J Bernstein}.}
  \bibinfo{year}{2006}\natexlab{}.
\newblock \showarticletitle{Curve25519: new Diffie-Hellman speed records}. In
  \bibinfo{booktitle}{\emph{9th International Conference on Theory and Practice
  in Public-Key Cryptography}}. \bibinfo{publisher}{Springer},
  \bibinfo{pages}{207--228}.
\newblock


\bibitem[Blum et~al\mbox{.}(2019)]%
        {blum2019non}
\bibfield{author}{\bibinfo{person}{Manuel Blum}, \bibinfo{person}{Paul
  Feldman}, {and} \bibinfo{person}{Silvio Micali}.}
  \bibinfo{year}{2019}\natexlab{}.
\newblock \showarticletitle{Non-interactive zero-knowledge and its
  applications}.
\newblock In \bibinfo{booktitle}{\emph{Providing Sound Foundations for
  Cryptography: On the Work of Shafi Goldwasser and Silvio Micali}},
  \bibfield{editor}{\bibinfo{person}{Oded Goldreich}} (Ed.).
  \bibinfo{pages}{329--349}.
\newblock


\bibitem[Botrel et~al\mbox{.}(2023)]%
        {gnark-v0.8.0}
\bibfield{author}{\bibinfo{person}{Gautam Botrel}, \bibinfo{person}{Thomas
  Piellard}, \bibinfo{person}{Youssef~El Housni}, \bibinfo{person}{Ivo Kubjas},
  {and} \bibinfo{person}{Arya Tabaie}.} \bibinfo{year}{2023}\natexlab{}.
\newblock \bibinfo{booktitle}{\emph{ConsenSys/gnark: v0.8.0}}.
\newblock
\urldef\tempurl%
\url{https://doi.org/10.5281/zenodo.5819104}
\showDOI{\tempurl}


\bibitem[Bowe et~al\mbox{.}(2017)]%
        {bowe2017scalable}
\bibfield{author}{\bibinfo{person}{Sean Bowe}, \bibinfo{person}{Ariel Gabizon},
  {and} \bibinfo{person}{Ian Miers}.} \bibinfo{year}{2017}\natexlab{}.
\newblock \showarticletitle{Scalable multi-party computation for zk-SNARK
  parameters in the random beacon model}.
\newblock \bibinfo{journal}{\emph{Cryptology ePrint Archive}}
  (\bibinfo{year}{2017}).
\newblock


\bibitem[Breidenbach et~al\mbox{.}(2021a)]%
        {breidenbach2021chainlink}
\bibfield{author}{\bibinfo{person}{Lorenz Breidenbach},
  \bibinfo{person}{Christian Cachin}, \bibinfo{person}{Benedict Chan},
  \bibinfo{person}{Alex Coventry}, \bibinfo{person}{Steve Ellis},
  \bibinfo{person}{Ari Juels}, \bibinfo{person}{Farinaz Koushanfar},
  \bibinfo{person}{Andrew Miller}, \bibinfo{person}{Brendan Magauran},
  \bibinfo{person}{Daniel Moroz}, {et~al\mbox{.}}}
  \bibinfo{year}{2021}\natexlab{a}.
\newblock \bibinfo{booktitle}{\emph{Chainlink 2.0: Next steps in the evolution
  of decentralized oracle networks}}.
\newblock
\urldef\tempurl%
\url{https://naorib.ir/white-paper/chinlink-whitepaper.pdf}
\showURL{%
Retrieved 2023-03-16 from \tempurl}


\bibitem[Breidenbach et~al\mbox{.}(2021b)]%
        {breidenbach2021ocr}
\bibfield{author}{\bibinfo{person}{Lorenz Breidenbach},
  \bibinfo{person}{Christian Cachin}, \bibinfo{person}{Alex Coventry},
  \bibinfo{person}{Ari Juels}, {and} \bibinfo{person}{Andrew Miller}.}
  \bibinfo{year}{2021}\natexlab{b}.
\newblock \bibinfo{booktitle}{\emph{Chainlink off-chain reporting protocol}}.
\newblock
\urldef\tempurl%
\url{https://research.chain.link/ocr.pdf}
\showURL{%
Retrieved 2024-03-07 from \tempurl}


\bibitem[B{\"u}nz et~al\mbox{.}(2018)]%
        {bunz2018bulletproofs}
\bibfield{author}{\bibinfo{person}{Benedikt B{\"u}nz},
  \bibinfo{person}{Jonathan Bootle}, \bibinfo{person}{Dan Boneh},
  \bibinfo{person}{Andrew Poelstra}, \bibinfo{person}{Pieter Wuille}, {and}
  \bibinfo{person}{Greg Maxwell}.} \bibinfo{year}{2018}\natexlab{}.
\newblock \showarticletitle{Bulletproofs: Short proofs for confidential
  transactions and more}. In \bibinfo{booktitle}{\emph{2018 IEEE Symposium on
  Security and Privacy}}. \bibinfo{publisher}{IEEE}, \bibinfo{pages}{315--334}.
\newblock


\bibitem[Buterin(2016)]%
        {buterin2016chain}
\bibfield{author}{\bibinfo{person}{Vitalik Buterin}.}
  \bibinfo{year}{2016}\natexlab{}.
\newblock \showarticletitle{Chain interoperability}.
\newblock \bibinfo{journal}{\emph{R3 Research Paper}}  \bibinfo{volume}{9}
  (\bibinfo{year}{2016}).
\newblock


\bibitem[Buterin(2021)]%
        {buterin2021rollups}
\bibfield{author}{\bibinfo{person}{Vitalik Buterin}.}
  \bibinfo{year}{2021}\natexlab{}.
\newblock \bibinfo{booktitle}{\emph{An Incomplete Guide to Rollups}}.
\newblock
\urldef\tempurl%
\url{https://vitalik.ca/general/2021/01/05/rollup.html}
\showURL{%
Retrieved 2023-03-09 from \tempurl}


\bibitem[Cai et~al\mbox{.}(2018)]%
        {cai2018decentralized}
\bibfield{author}{\bibinfo{person}{Wei Cai}, \bibinfo{person}{Zehua Wang},
  \bibinfo{person}{Jason~B Ernst}, \bibinfo{person}{Zhen Hong},
  \bibinfo{person}{Chen Feng}, {and} \bibinfo{person}{Victor~CM Leung}.}
  \bibinfo{year}{2018}\natexlab{}.
\newblock \showarticletitle{Decentralized applications: The
  blockchain-empowered software system}.
\newblock \bibinfo{journal}{\emph{IEEE access}}  \bibinfo{volume}{6}
  (\bibinfo{year}{2018}), \bibinfo{pages}{53019--53033}.
\newblock


\bibitem[Chiesa et~al\mbox{.}(2020)]%
        {chiesa2020marlin}
\bibfield{author}{\bibinfo{person}{Alessandro Chiesa}, \bibinfo{person}{Yuncong
  Hu}, \bibinfo{person}{Mary Maller}, \bibinfo{person}{Pratyush Mishra},
  \bibinfo{person}{Noah Vesely}, {and} \bibinfo{person}{Nicholas Ward}.}
  \bibinfo{year}{2020}\natexlab{}.
\newblock \showarticletitle{Marlin: Preprocessing zkSNARKs with universal and
  updatable SRS}. In \bibinfo{booktitle}{\emph{39th Annual International
  Conference on the Theory and Applications of Cryptographic Techniques}}.
  \bibinfo{publisher}{Springer}, \bibinfo{pages}{738--768}.
\newblock


\bibitem[Costello et~al\mbox{.}(2015)]%
        {costello2015geppetto}
\bibfield{author}{\bibinfo{person}{Craig Costello}, \bibinfo{person}{C{\'e}dric
  Fournet}, \bibinfo{person}{Jon Howell}, \bibinfo{person}{Markulf Kohlweiss},
  \bibinfo{person}{Benjamin Kreuter}, \bibinfo{person}{Michael Naehrig},
  \bibinfo{person}{Bryan Parno}, {and} \bibinfo{person}{Samee Zahur}.}
  \bibinfo{year}{2015}\natexlab{}.
\newblock \showarticletitle{Geppetto: Versatile verifiable computation}. In
  \bibinfo{booktitle}{\emph{2015 IEEE Symposium on Security and Privacy}}.
  IEEE, \bibinfo{pages}{253--270}.
\newblock


\bibitem[Dai et~al\mbox{.}(2019)]%
        {dai2019blockchain}
\bibfield{author}{\bibinfo{person}{Hong-Ning Dai}, \bibinfo{person}{Zibin
  Zheng}, {and} \bibinfo{person}{Yan Zhang}.} \bibinfo{year}{2019}\natexlab{}.
\newblock \showarticletitle{Blockchain for Internet of Things: A survey}.
\newblock \bibinfo{journal}{\emph{IEEE Internet of Things Journal}}
  \bibinfo{volume}{6}, \bibinfo{number}{5} (\bibinfo{year}{2019}),
  \bibinfo{pages}{8076--8094}.
\newblock


\bibitem[Demestichas et~al\mbox{.}(2020)]%
        {demestichas2020blockchain}
\bibfield{author}{\bibinfo{person}{Konstantinos Demestichas},
  \bibinfo{person}{Nikolaos Peppes}, \bibinfo{person}{Theodoros Alexakis},
  {and} \bibinfo{person}{Evgenia Adamopoulou}.}
  \bibinfo{year}{2020}\natexlab{}.
\newblock \showarticletitle{Blockchain in agriculture traceability systems: A
  review}.
\newblock \bibinfo{journal}{\emph{Applied Sciences}} \bibinfo{volume}{10},
  \bibinfo{number}{12} (\bibinfo{year}{2020}), \bibinfo{pages}{4113}.
\newblock


\bibitem[Douceur(2002)]%
        {douceur2002sybil}
\bibfield{author}{\bibinfo{person}{John~R Douceur}.}
  \bibinfo{year}{2002}\natexlab{}.
\newblock \showarticletitle{The sybil attack}. In
  \bibinfo{booktitle}{\emph{International Workshop on Peer-to-Peer Systems}}.
  \bibinfo{publisher}{Springer}, \bibinfo{pages}{251--260}.
\newblock


\bibitem[Ellul and Pace(2022)]%
        {ellul2022verifiable}
\bibfield{author}{\bibinfo{person}{Joshua Ellul} {and}
  \bibinfo{person}{Gordon~J Pace}.} \bibinfo{year}{2022}\natexlab{}.
\newblock \showarticletitle{Verifiable External Blockchain Calls: Towards
  Removing Oracle Input Intermediaries}. In
  \bibinfo{booktitle}{\emph{International Workshop on Data Privacy
  Management}}. Springer, \bibinfo{pages}{317--324}.
\newblock


\bibitem[Foundation(2022)]%
        {hyperledger2022cacti}
\bibfield{author}{\bibinfo{person}{Hyperledger Foundation}.}
  \bibinfo{year}{2022}\natexlab{}.
\newblock \bibinfo{booktitle}{\emph{Introducing Hyperledger Cacti, a
  Multi-Faceted Pluggable Interoperability Framework}}.
\newblock
\urldef\tempurl%
\url{https://www.hyperledger.org/blog/2022/11/07/introducing-hyperledger-cacti-a-multi-faceted-pluggable-interoperability-framework}
\showURL{%
Retrieved 2024-03-07 from \tempurl}


\bibitem[Frauenthaler et~al\mbox{.}(2020)]%
        {frauenthaler2020eth}
\bibfield{author}{\bibinfo{person}{Philipp Frauenthaler},
  \bibinfo{person}{Marten Sigwart}, \bibinfo{person}{Christof Spanring},
  \bibinfo{person}{Michael Sober}, {and} \bibinfo{person}{Stefan Schulte}.}
  \bibinfo{year}{2020}\natexlab{}.
\newblock \showarticletitle{ETH relay: A cost-efficient relay for
  ethereum-based blockchains}. In \bibinfo{booktitle}{\emph{2020 IEEE
  International Conference on Blockchain}}. \bibinfo{publisher}{IEEE},
  \bibinfo{pages}{204--213}.
\newblock


\bibitem[Gabizon et~al\mbox{.}(2019)]%
        {gabizon2019plonk}
\bibfield{author}{\bibinfo{person}{Ariel Gabizon}, \bibinfo{person}{Zachary~J
  Williamson}, {and} \bibinfo{person}{Oana Ciobotaru}.}
  \bibinfo{year}{2019}\natexlab{}.
\newblock \showarticletitle{Plonk: Permutations over lagrange-bases for
  oecumenical noninteractive arguments of knowledge}.
\newblock \bibinfo{journal}{\emph{Cryptology ePrint Archive}}
  (\bibinfo{year}{2019}).
\newblock


\bibitem[Garoffolo et~al\mbox{.}(2020)]%
        {garoffolo2020zendoo}
\bibfield{author}{\bibinfo{person}{Alberto Garoffolo}, \bibinfo{person}{Dmytro
  Kaidalov}, {and} \bibinfo{person}{Roman Oliynykov}.}
  \bibinfo{year}{2020}\natexlab{}.
\newblock \showarticletitle{Zendoo: A zk-SNARK verifiable cross-chain transfer
  protocol enabling decoupled and decentralized sidechains}. In
  \bibinfo{booktitle}{\emph{2020 IEEE 40th International Conference on
  Distributed Computing Systems (ICDCS)}}. IEEE, \bibinfo{pages}{1257--1262}.
\newblock


\bibitem[Gennaro et~al\mbox{.}(2010)]%
        {gennaro2010non}
\bibfield{author}{\bibinfo{person}{Rosario Gennaro}, \bibinfo{person}{Craig
  Gentry}, {and} \bibinfo{person}{Bryan Parno}.}
  \bibinfo{year}{2010}\natexlab{}.
\newblock \showarticletitle{Non-interactive verifiable computing: Outsourcing
  computation to untrusted workers}. In \bibinfo{booktitle}{\emph{30th Annual
  Cryptology Conference}}. \bibinfo{publisher}{Springer},
  \bibinfo{pages}{465--482}.
\newblock


\bibitem[Goes(2020)]%
        {goes2020interblockchain}
\bibfield{author}{\bibinfo{person}{Christopher Goes}.}
  \bibinfo{year}{2020}\natexlab{}.
\newblock \showarticletitle{The Interblockchain Communication Protocol: An
  Overview}.
\newblock \bibinfo{journal}{\emph{arXiv preprint arXiv:2006.15918}}
  (\bibinfo{year}{2020}).
\newblock


\bibitem[Goldreich et~al\mbox{.}(1991)]%
        {goldreich1991proofs}
\bibfield{author}{\bibinfo{person}{Oded Goldreich}, \bibinfo{person}{Silvio
  Micali}, {and} \bibinfo{person}{Avi Wigderson}.}
  \bibinfo{year}{1991}\natexlab{}.
\newblock \showarticletitle{Proofs that yield nothing but their validity or all
  languages in NP have zero-knowledge proof systems}.
\newblock \bibinfo{journal}{\emph{J. ACM}} \bibinfo{volume}{38},
  \bibinfo{number}{3} (\bibinfo{year}{1991}), \bibinfo{pages}{690--728}.
\newblock


\bibitem[Goldwasser et~al\mbox{.}(2019)]%
        {goldwasser2019knowledge}
\bibfield{author}{\bibinfo{person}{Shafi Goldwasser}, \bibinfo{person}{Silvio
  Micali}, {and} \bibinfo{person}{Chales Rackoff}.}
  \bibinfo{year}{2019}\natexlab{}.
\newblock \showarticletitle{The knowledge complexity of interactive
  proof-systems}.
\newblock In \bibinfo{booktitle}{\emph{Providing Sound Foundations for
  Cryptography: On the Work of Shafi Goldwasser and Silvio Micali}},
  \bibfield{editor}{\bibinfo{person}{Oded Goldreich}} (Ed.).
  \bibinfo{pages}{203--225}.
\newblock


\bibitem[Goldwasser et~al\mbox{.}(1988)]%
        {goldwasser1988digital}
\bibfield{author}{\bibinfo{person}{Shafi Goldwasser}, \bibinfo{person}{Silvio
  Micali}, {and} \bibinfo{person}{Ronald~L Rivest}.}
  \bibinfo{year}{1988}\natexlab{}.
\newblock \showarticletitle{A digital signature scheme secure against adaptive
  chosen-message attacks}.
\newblock \bibinfo{journal}{\emph{SIAM Journal on computing}}
  \bibinfo{volume}{17}, \bibinfo{number}{2} (\bibinfo{year}{1988}),
  \bibinfo{pages}{281--308}.
\newblock


\bibitem[Groth(2016)]%
        {groth2016size}
\bibfield{author}{\bibinfo{person}{Jens Groth}.}
  \bibinfo{year}{2016}\natexlab{}.
\newblock \showarticletitle{On the size of pairing-based non-interactive
  arguments}. In \bibinfo{booktitle}{\emph{Annual International Conference on
  the Theory and Applications of Cryptographic Techniques}}.
  \bibinfo{publisher}{Springer}, \bibinfo{pages}{305--326}.
\newblock


\bibitem[Hafid et~al\mbox{.}(2020)]%
        {hafid2020scaling}
\bibfield{author}{\bibinfo{person}{Abdelatif Hafid},
  \bibinfo{person}{Abdelhakim~Senhaji Hafid}, {and} \bibinfo{person}{Mustapha
  Samih}.} \bibinfo{year}{2020}\natexlab{}.
\newblock \showarticletitle{Scaling blockchains: A comprehensive survey}.
\newblock \bibinfo{journal}{\emph{IEEE Access}}  \bibinfo{volume}{8}
  (\bibinfo{year}{2020}), \bibinfo{pages}{125244--125262}.
\newblock


\bibitem[Heiss et~al\mbox{.}(2019)]%
        {heiss2019oracles}
\bibfield{author}{\bibinfo{person}{Jonathan Heiss}, \bibinfo{person}{Jacob
  Eberhardt}, {and} \bibinfo{person}{Stefan Tai}.}
  \bibinfo{year}{2019}\natexlab{}.
\newblock \showarticletitle{From oracles to trustworthy data on-chaining
  systems}. In \bibinfo{booktitle}{\emph{2019 IEEE International Conference on
  Blockchain (Blockchain)}}. IEEE, \bibinfo{pages}{496--503}.
\newblock


\bibitem[Herlihy(2018)]%
        {herlihy2018atomic}
\bibfield{author}{\bibinfo{person}{Maurice Herlihy}.}
  \bibinfo{year}{2018}\natexlab{}.
\newblock \showarticletitle{Atomic cross-chain swaps}. In
  \bibinfo{booktitle}{\emph{2018 ACM Symposium on Principles of Distributed
  Computing}}. \bibinfo{publisher}{ACM}, \bibinfo{pages}{245--254}.
\newblock


\bibitem[Hopwood et~al\mbox{.}(2016)]%
        {hopwood2016zcash}
\bibfield{author}{\bibinfo{person}{Daira Hopwood}, \bibinfo{person}{Sean Bowe},
  \bibinfo{person}{Taylor Hornby}, {and} \bibinfo{person}{Nathan Wilcox}.}
  \bibinfo{year}{2016}\natexlab{}.
\newblock \bibinfo{booktitle}{\emph{Zcash Protocol Specication}}.
\newblock
\urldef\tempurl%
\url{https://raw.githubusercontent.com/zcash/zips/master/protocol/protocol.pdf}
\showURL{%
Retrieved 2022-06-24 from \tempurl}


\bibitem[Kalodner et~al\mbox{.}(2018)]%
        {kalodner2018arbitrum}
\bibfield{author}{\bibinfo{person}{Harry Kalodner}, \bibinfo{person}{Steven
  Goldfeder}, \bibinfo{person}{Xiaoqi Chen}, \bibinfo{person}{S~Matthew
  Weinberg}, {and} \bibinfo{person}{Edward~W Felten}.}
  \bibinfo{year}{2018}\natexlab{}.
\newblock \showarticletitle{Arbitrum: Scalable, private smart contracts}. In
  \bibinfo{booktitle}{\emph{27th USENIX Security Symposium}}.
  \bibinfo{publisher}{USENIX Association}, \bibinfo{pages}{1353--1370}.
\newblock


\bibitem[Kwon and Buchman(2020)]%
        {kwon2020cosmos}
\bibfield{author}{\bibinfo{person}{Jae Kwon} {and} \bibinfo{person}{Ethan
  Buchman}.} \bibinfo{year}{2020}\natexlab{}.
\newblock \bibinfo{booktitle}{\emph{{Cosmos} Whitepaper: A Network of
  Distributed Ledgers}}.
\newblock
\urldef\tempurl%
\url{https://wikibitimg.fx994.com/attach/2020/12/16623142020/WBE16623142020_55300.pdf}
\showURL{%
Retrieved 2023-03-16 from \tempurl}


\bibitem[Lafourcade and Lombard-Platet(2020)]%
        {lafourcade2020blockchain}
\bibfield{author}{\bibinfo{person}{Pascal Lafourcade} {and}
  \bibinfo{person}{Marius Lombard-Platet}.} \bibinfo{year}{2020}\natexlab{}.
\newblock \showarticletitle{About blockchain interoperability}.
\newblock \bibinfo{journal}{\emph{Inform. Process. Lett.}}
  \bibinfo{volume}{161} (\bibinfo{year}{2020}), \bibinfo{pages}{105976}.
\newblock


\bibitem[Liu et~al\mbox{.}(2019)]%
        {liu2019hyperservice}
\bibfield{author}{\bibinfo{person}{Zhuotao Liu}, \bibinfo{person}{Yangxi
  Xiang}, \bibinfo{person}{Jian Shi}, \bibinfo{person}{Peng Gao},
  \bibinfo{person}{Haoyu Wang}, \bibinfo{person}{Xusheng Xiao},
  \bibinfo{person}{Bihan Wen}, {and} \bibinfo{person}{Yih-Chun Hu}.}
  \bibinfo{year}{2019}\natexlab{}.
\newblock \showarticletitle{Hyperservice: Interoperability and programmability
  across heterogeneous blockchains}. In \bibinfo{booktitle}{\emph{2019 ACM
  SIGSAC Conference on Computer and Communications Security}}.
  \bibinfo{publisher}{ACM}, \bibinfo{pages}{549--566}.
\newblock


\bibitem[Madine et~al\mbox{.}(2021)]%
        {madine2021appxchain}
\bibfield{author}{\bibinfo{person}{Mohammad Madine}, \bibinfo{person}{Khaled
  Salah}, \bibinfo{person}{Raja Jayaraman}, \bibinfo{person}{Yousof
  Al-Hammadi}, \bibinfo{person}{Junaid Arshad}, {and} \bibinfo{person}{Ibrar
  Yaqoob}.} \bibinfo{year}{2021}\natexlab{}.
\newblock \showarticletitle{appxchain: Application-level interoperability for
  blockchain networks}.
\newblock \bibinfo{journal}{\emph{IEEE Access}}  \bibinfo{volume}{9}
  (\bibinfo{year}{2021}), \bibinfo{pages}{87777--87791}.
\newblock


\bibitem[Maller et~al\mbox{.}(2019)]%
        {maller2019sonic}
\bibfield{author}{\bibinfo{person}{Mary Maller}, \bibinfo{person}{Sean Bowe},
  \bibinfo{person}{Markulf Kohlweiss}, {and} \bibinfo{person}{Sarah
  Meiklejohn}.} \bibinfo{year}{2019}\natexlab{}.
\newblock \showarticletitle{Sonic: Zero-knowledge SNARKs from linear-size
  universal and updatable structured reference strings}. In
  \bibinfo{booktitle}{\emph{2019 ACM SIGSAC Conference on Computer and
  Communications Security}}. \bibinfo{publisher}{ACM},
  \bibinfo{pages}{2111--2128}.
\newblock


\bibitem[Nakamoto(2008)]%
        {nakamoto2008bitcoin}
\bibfield{author}{\bibinfo{person}{Satoshi Nakamoto}.}
  \bibinfo{year}{2008}\natexlab{}.
\newblock \bibinfo{booktitle}{\emph{Bitcoin: A peer-to-peer electronic cash
  system}}.
\newblock
\urldef\tempurl%
\url{http://www.bitcoin.org/bitcoin.pdf}
\showURL{%
Retrieved 2022-07-02 from \tempurl}


\bibitem[Nissl et~al\mbox{.}(2021)]%
        {nissl2021crossblockchain}
\bibfield{author}{\bibinfo{person}{Markus Nissl}, \bibinfo{person}{Emanuel
  Sallinger}, \bibinfo{person}{Stefan Schulte}, {and} \bibinfo{person}{Michael
  Borkowski}.} \bibinfo{year}{2021}\natexlab{}.
\newblock \showarticletitle{Towards Cross-Blockchain Smart Contracts}. In
  \bibinfo{booktitle}{\emph{2021 IEEE International Conference on Decentralized
  Applications and Infrastructures}}. \bibinfo{publisher}{IEEE},
  \bibinfo{pages}{85--94}.
\newblock


\bibitem[Parno et~al\mbox{.}(2016)]%
        {parno2016pinocchio}
\bibfield{author}{\bibinfo{person}{Bryan Parno}, \bibinfo{person}{Jon Howell},
  \bibinfo{person}{Craig Gentry}, {and} \bibinfo{person}{Mariana Raykova}.}
  \bibinfo{year}{2016}\natexlab{}.
\newblock \showarticletitle{Pinocchio: Nearly practical verifiable
  computation}.
\newblock \bibinfo{journal}{\emph{Commun. ACM}} \bibinfo{volume}{59},
  \bibinfo{number}{2} (\bibinfo{year}{2016}), \bibinfo{pages}{103--112}.
\newblock


\bibitem[Parno et~al\mbox{.}(2012)]%
        {parno2012delegate}
\bibfield{author}{\bibinfo{person}{Bryan Parno}, \bibinfo{person}{Mariana
  Raykova}, {and} \bibinfo{person}{Vinod Vaikuntanathan}.}
  \bibinfo{year}{2012}\natexlab{}.
\newblock \showarticletitle{How to delegate and verify in public: Verifiable
  computation from attribute-based encryption}. In
  \bibinfo{booktitle}{\emph{9th Theory of Cryptography Conference}}.
  \bibinfo{publisher}{Springer}, \bibinfo{pages}{422--439}.
\newblock


\bibitem[Robinson and Ramesh(2021)]%
        {gpact2021robinson}
\bibfield{author}{\bibinfo{person}{Peter Robinson} {and}
  \bibinfo{person}{Raghavendra Ramesh}.} \bibinfo{year}{2021}\natexlab{}.
\newblock \showarticletitle{General Purpose Atomic Crosschain Transactions}. In
  \bibinfo{booktitle}{\emph{2021 3rd Conference on Blockchain Research
  Applications for Innovative Networks and Services}}.
  \bibinfo{publisher}{IEEE}, \bibinfo{pages}{61--68}.
\newblock


\bibitem[Saberi et~al\mbox{.}(2019)]%
        {saberi2019blockchain}
\bibfield{author}{\bibinfo{person}{Sara Saberi}, \bibinfo{person}{Mahtab
  Kouhizadeh}, \bibinfo{person}{Joseph Sarkis}, {and} \bibinfo{person}{Lejia
  Shen}.} \bibinfo{year}{2019}\natexlab{}.
\newblock \showarticletitle{Blockchain technology and its relationships to
  sustainable supply chain management}.
\newblock \bibinfo{journal}{\emph{International Journal of Production
  Research}} \bibinfo{volume}{57}, \bibinfo{number}{7} (\bibinfo{year}{2019}),
  \bibinfo{pages}{2117--2135}.
\newblock


\bibitem[Scaffino et~al\mbox{.}(2023)]%
        {scaffino2023glimpse}
\bibfield{author}{\bibinfo{person}{Giulia Scaffino}, \bibinfo{person}{Lukas
  Aumayr}, \bibinfo{person}{Zeta Avarikioti}, {and} \bibinfo{person}{Matteo
  Maffei}.} \bibinfo{year}{2023}\natexlab{}.
\newblock \showarticletitle{Glimpse: On-Demand PoW Light Client with
  Constant-Size Storage for DeFi}. In \bibinfo{booktitle}{\emph{32nd USENIX
  Security Symposium (USENIX Security 23)}}. \bibinfo{publisher}{USENIX
  Association}.
\newblock


\bibitem[Scaffino et~al\mbox{.}(2024)]%
        {scaffino2024alba}
\bibfield{author}{\bibinfo{person}{Giulia Scaffino}, \bibinfo{person}{Lukas
  Aumayr}, \bibinfo{person}{Mahsa Bastankhah}, \bibinfo{person}{Zeta
  Avarikioti}, {and} \bibinfo{person}{Matteo Maffei}.}
  \bibinfo{year}{2024}\natexlab{}.
\newblock \bibinfo{title}{Alba: The Dawn of Scalable Bridges for Blockchains}.
\newblock \bibinfo{howpublished}{Cryptology ePrint Archive, Paper 2024/197}.
\newblock
\urldef\tempurl%
\url{https://eprint.iacr.org/2024/197}
\showURL{%
\tempurl}


\bibitem[Schulte et~al\mbox{.}(2019)]%
        {schulte2019towards}
\bibfield{author}{\bibinfo{person}{Stefan Schulte}, \bibinfo{person}{Marten
  Sigwart}, \bibinfo{person}{Philipp Frauenthaler}, {and}
  \bibinfo{person}{Michael Borkowski}.} \bibinfo{year}{2019}\natexlab{}.
\newblock \showarticletitle{Towards blockchain interoperability}. In
  \bibinfo{booktitle}{\emph{Business Process Management: Blockchain and Central
  and Eastern Europe Forum: BPM 2019 Blockchain and CEE Forum}}.
  \bibinfo{publisher}{Springer}, \bibinfo{pages}{3--10}.
\newblock


\bibitem[Sober et~al\mbox{.}(2023)]%
        {sober2023distributed}
\bibfield{author}{\bibinfo{person}{Michael Sober}, \bibinfo{person}{Max
  Kobelt}, \bibinfo{person}{Giulia Scaffino}, \bibinfo{person}{Dominik Kaaser},
  {and} \bibinfo{person}{Stefan Schulte}.} \bibinfo{year}{2023}\natexlab{}.
\newblock \showarticletitle{Distributed Key Generation with Smart Contracts
  using zk-SNARKs}. In \bibinfo{booktitle}{\emph{38th ACM/SIGAPP Symposium on
  Applied Computing (SAC '23)}}. \bibinfo{publisher}{ACM}.
\newblock
\newblock
\shownote{In Press}.


\bibitem[Sober et~al\mbox{.}(2021)]%
        {sober2021voting}
\bibfield{author}{\bibinfo{person}{Michael Sober}, \bibinfo{person}{Giulia
  Scaffino}, \bibinfo{person}{Christof Spanring}, {and} \bibinfo{person}{Stefan
  Schulte}.} \bibinfo{year}{2021}\natexlab{}.
\newblock \showarticletitle{A voting-based blockchain interoperability oracle}.
  In \bibinfo{booktitle}{\emph{2021 IEEE International Conference on
  Blockchain}}. IEEE, \bibinfo{pages}{160--169}.
\newblock


\bibitem[Thibault et~al\mbox{.}(2022)]%
        {thibault2022rollups}
\bibfield{author}{\bibinfo{person}{Louis~Tremblay Thibault},
  \bibinfo{person}{Tom Sarry}, {and} \bibinfo{person}{Abdelhakim~Senhaji
  Hafid}.} \bibinfo{year}{2022}\natexlab{}.
\newblock \showarticletitle{Blockchain Scaling Using Rollups: {A} Comprehensive
  Survey}.
\newblock \bibinfo{journal}{\emph{{IEEE} Access}}  \bibinfo{volume}{10}
  (\bibinfo{year}{2022}), \bibinfo{pages}{93039--93054}.
\newblock


\bibitem[Wang et~al\mbox{.}(2023)]%
        {wang2023interoperability}
\bibfield{author}{\bibinfo{person}{Gang Wang}, \bibinfo{person}{Qin Wang},
  {and} \bibinfo{person}{Shiping Chen}.} \bibinfo{year}{2023}\natexlab{}.
\newblock \showarticletitle{Exploring Blockchains Interoperability: A
  Systematic Survey}.
\newblock \bibinfo{journal}{\emph{ACM Comput. Surv.}} (\bibinfo{year}{2023}).
\newblock


\bibitem[Westerkamp and Diez(2022)]%
        {westerkamp2022verilay}
\bibfield{author}{\bibinfo{person}{Martin Westerkamp} {and}
  \bibinfo{person}{Maximilian Diez}.} \bibinfo{year}{2022}\natexlab{}.
\newblock \showarticletitle{Verilay: A verifiable proof of stake chain relay}.
  In \bibinfo{booktitle}{\emph{2022 IEEE International Conference on Blockchain
  and Cryptocurrency}}. \bibinfo{publisher}{IEEE}, \bibinfo{pages}{1--9}.
\newblock


\bibitem[Westerkamp and Eberhardt(2020)]%
        {westerkamp2020zkrelay}
\bibfield{author}{\bibinfo{person}{Martin Westerkamp} {and}
  \bibinfo{person}{Jacob Eberhardt}.} \bibinfo{year}{2020}\natexlab{}.
\newblock \showarticletitle{zkrelay: Facilitating sidechains using
  zksnark-based chain-relays}. In \bibinfo{booktitle}{\emph{2020 IEEE European
  Symposium on Security and Privacy Workshops}}. \bibinfo{publisher}{IEEE},
  \bibinfo{pages}{378--386}.
\newblock


\bibitem[Wood(2016)]%
        {wood2016polkadot}
\bibfield{author}{\bibinfo{person}{Gavin Wood}.}
  \bibinfo{year}{2016}\natexlab{}.
\newblock \bibinfo{booktitle}{\emph{Polkadot: Vision for a heterogeneous
  multi-chain framework}}.
\newblock
\urldef\tempurl%
\url{https://www.win.tue.nl/~mholende/seminar/references/ethereum_polkadot.pdf}
\showURL{%
Retrieved 2023-03-16 from \tempurl}


\bibitem[Wood et~al\mbox{.}(2014)]%
        {wood2014ethereum}
\bibfield{author}{\bibinfo{person}{Gavin Wood} {et~al\mbox{.}}}
  \bibinfo{year}{2014}\natexlab{}.
\newblock \showarticletitle{Ethereum: A secure decentralised generalised
  transaction ledger}.
\newblock \bibinfo{journal}{\emph{Ethereum project yellow paper}}
  \bibinfo{volume}{151}, \bibinfo{number}{2014} (\bibinfo{year}{2014}),
  \bibinfo{pages}{1--32}.
\newblock


\bibitem[Xie et~al\mbox{.}(2019)]%
        {xie2019survey}
\bibfield{author}{\bibinfo{person}{Junfeng Xie}, \bibinfo{person}{Helen Tang},
  \bibinfo{person}{Tao Huang}, \bibinfo{person}{F~Richard Yu},
  \bibinfo{person}{Renchao Xie}, \bibinfo{person}{Jiang Liu}, {and}
  \bibinfo{person}{Yunjie Liu}.} \bibinfo{year}{2019}\natexlab{}.
\newblock \showarticletitle{A survey of blockchain technology applied to smart
  cities: Research issues and challenges}.
\newblock \bibinfo{journal}{\emph{IEEE Communications Surveys \& Tutorials}}
  \bibinfo{volume}{21}, \bibinfo{number}{3} (\bibinfo{year}{2019}),
  \bibinfo{pages}{2794--2830}.
\newblock


\bibitem[Xie et~al\mbox{.}(2022)]%
        {xie2022zkbridge}
\bibfield{author}{\bibinfo{person}{Tiancheng Xie}, \bibinfo{person}{Jiaheng
  Zhang}, \bibinfo{person}{Zerui Cheng}, \bibinfo{person}{Fan Zhang},
  \bibinfo{person}{Yupeng Zhang}, \bibinfo{person}{Yongzheng Jia},
  \bibinfo{person}{Dan Boneh}, {and} \bibinfo{person}{Dawn Song}.}
  \bibinfo{year}{2022}\natexlab{}.
\newblock \showarticletitle{zkbridge: Trustless cross-chain bridges made
  practical}.
\newblock \bibinfo{journal}{\emph{arXiv preprint arXiv:2210.00264}}
  (\bibinfo{year}{2022}).
\newblock


\bibitem[Yeh et~al\mbox{.}(2022)]%
        {yeh2022secure}
\bibfield{author}{\bibinfo{person}{Kuo-Hui Yeh}, \bibinfo{person}{Guan-Yan
  Yang}, \bibinfo{person}{Chanapha Butpheng}, \bibinfo{person}{Lin-Fa Lee},
  {and} \bibinfo{person}{Ying-Ho Liu}.} \bibinfo{year}{2022}\natexlab{}.
\newblock \showarticletitle{A Secure Interoperability Management Scheme for
  Cross-Blockchain Transactions}.
\newblock \bibinfo{journal}{\emph{Symmetry}} \bibinfo{volume}{14},
  \bibinfo{number}{12} (\bibinfo{year}{2022}), \bibinfo{pages}{2473}.
\newblock


\bibitem[Zamyatin et~al\mbox{.}(2021)]%
        {zamyatin2021sok}
\bibfield{author}{\bibinfo{person}{Alexei Zamyatin}, \bibinfo{person}{Mustafa
  Al-Bassam}, \bibinfo{person}{Dionysis Zindros}, \bibinfo{person}{Eleftherios
  Kokoris-Kogias}, \bibinfo{person}{Pedro Moreno-Sanchez},
  \bibinfo{person}{Aggelos Kiayias}, {and} \bibinfo{person}{William~J
  Knottenbelt}.} \bibinfo{year}{2021}\natexlab{}.
\newblock \showarticletitle{SoK: Communication across distributed ledgers}. In
  \bibinfo{booktitle}{\emph{25th International Conference on Financial
  Cryptography and Data Security, Revised Selected Papers, Part II}}.
  \bibinfo{publisher}{Springer}, \bibinfo{pages}{3--36}.
\newblock


\bibitem[Zhang et~al\mbox{.}(2020)]%
        {zhang2020heuristic}
\bibfield{author}{\bibinfo{person}{Yuhang Zhang}, \bibinfo{person}{Jun Wang},
  {and} \bibinfo{person}{Jie Luo}.} \bibinfo{year}{2020}\natexlab{}.
\newblock \showarticletitle{Heuristic-based address clustering in bitcoin}.
\newblock \bibinfo{journal}{\emph{IEEE Access}}  \bibinfo{volume}{8}
  (\bibinfo{year}{2020}), \bibinfo{pages}{210582--210591}.
\newblock


\bibitem[Zhou et~al\mbox{.}(2020)]%
        {zhou2020solutions}
\bibfield{author}{\bibinfo{person}{Qiheng Zhou}, \bibinfo{person}{Huawei
  Huang}, \bibinfo{person}{Zibin Zheng}, {and} \bibinfo{person}{Jing Bian}.}
  \bibinfo{year}{2020}\natexlab{}.
\newblock \showarticletitle{Solutions to scalability of blockchain: A survey}.
\newblock \bibinfo{journal}{\emph{Ieee Access}}  \bibinfo{volume}{8}
  (\bibinfo{year}{2020}), \bibinfo{pages}{16440--16455}.
\newblock


\end{thebibliography}
